%% file: tmst.tex
\documentclass[a4paper,UKenglish,numberwithinsect]{lipics}
\usepackage{amsmath,amssymb,url,tikz}
\usepackage{graphicx}

\newcommand{\nats}{\mathbb{N}}

\def\poly{\mathrm{poly}}
\def\eps{\varepsilon}

\def\DST{\mathrm{DST}}
\def\DSF{\mathrm{DSF}}
\def\rTC{r\mbox{-}\mathrm{MTC}}
\def\TC{\mathrm{MTC}}
\def\SLC{\mathrm{SLC}}
\def\NP{\mathrm{NP}}
\def\ZTIME{\mathrm{ZTIME}}
\def\DTIME{\mathrm{DTIME}}
\def\APX{\mathrm{APX}}

\def\G{\mathcal{G}}

\title{On the Size and the Approximability of Minimum Temporally Connected Subgraphs}
\titlerunning{Size and Approximability of Minimum Temporally Connected Subgraphs}

\author{Kyriakos Axiotis}
\author{Dimitris Fotakis}

\affil{School of Electrical and Computer Engineering, National Technical University of Athens, Greece \\
\texttt{kaxiotis@corelab.ntua.gr}, \texttt{fotakis@cs.ntua.gr}}

\authorrunning{K. Axiotis and D. Fotakis}

\Copyright{Kyriakos Axiotis and Dimitris Fotakis}


\serieslogo{}
\volumeinfo
  {Billy Editor, Bill Editors}
  {2}
  {Conference title on which this volume is based on}
  {1}
  {1}
  {1}
\EventShortName{}
\DOI{10.4230/LIPIcs.xxx.yyy.p}

\begin{document}

\maketitle

\begin{abstract}
We consider \emph{temporal graphs} with discrete time labels and investigate the size and the approximability of minimum temporally connected spanning subgraphs. We present a family of minimally connected temporal graphs with $n$ vertices and $\Omega(n^2)$ edges, thus resolving an open question of (Kempe, Kleinberg, Kumar, JCSS~64, 2002) about the existence of sparse temporal connectivity certificates. Next, we consider the problem of computing a minimum weight subset of temporal edges that preserve connectivity of a given temporal graph either from a given vertex $r$ ($\rTC$ problem) or among all vertex pairs ($\TC$ problem). We show that the approximability of $\rTC$ is closely related to the approximability of Directed Steiner Tree and that $\rTC$ can be solved in polynomial time if the underlying graph has bounded treewidth. We also show that the best approximation ratio for $\TC$ is at least $O(2^{\log^{1-\eps} n})$ and at most $O(\min\{n^{1+\eps}, (\Delta M)^{2/3+\eps}\})$, for any constant $\eps > 0$, where $M$ is the number of temporal edges and $\Delta$ is the maximum degree of the underlying graph. Furthermore, we prove that the unweighted version of $\TC$ is $\APX$-hard and that $\TC$ is efficiently solvable in trees and $2$-approximable in cycles.
\end{abstract}


\input{intro}

\input{prelim}

\input{lower}

\input{rooted_mtcs}

\input{mtcs}

\input{special}

\bibliographystyle{plain}
\bibliography{temporal}

\appendix\section*{Appendix}
\makeatletter
\edef\thetheorem{\expandafter\noexpand\thesection\@thmcountersep\@thmcounter{theorem}}
\makeatother
\input{appendix_lower}
\input{appendix_rooted}

\input{appendix_mtcs}

\input{appendix_special}
\end{document}

%% file: intro.tex
\section{Introduction}
\label{s:intro}

Graphs and networks are ubiquitous in Computer Science, as they provide a natural and useful abstraction of many complex systems (e.g., transportation and communication networks) and processes (e.g., information spreading, epidemics, routing), and also of the interaction between individual entities or particles (e.g., social networks, chemical and biological networks). Traditional graph theoretic models assume that the structure of the network and the strength of interaction are time-invariant. However, as observed in e.g., \cite{Ber96,KKK00}, in many applications of graph theoretic models, the availability and the weights of the edges are actually time-dependent. For instance, one may think of information spreading and distributed computation in dynamic networks (see e.g., \cite{CFQS12,DPRSV13,KKK00,KLO10}), of mobile adhoc and sensor networks (see e.g., \cite{DW05}), of transportation networks and route planning (see e.g., \cite{Ber96,FHS11}), of epidemics, biological and ecological networks (see e.g., \cite{HS12,KKK00}), and of influence systems and coevolutionary opinion formation (see e.g., \cite{BGM13,Chaz12}).

Several variants of time-dependent graphs have been suggested as abstractions of such settings and computational problems (see e.g., \cite{CFQS12} and the references therein). No matter the particular variant, the main research questions are usually related either to optimizing or exploiting temporal connectivity or to computing short time-respecting paths (see e.g., \cite{AGMS15,AGMS14,Ber96,EHK15,FHS11,KKK00,MMCS13}). In this work, we adopt the simple and natural model of temporal graphs with discrete time labels \cite{KKK00} (and its extension with multiple labels per edge \cite{MMCS13}), and study the existence of dense minimally connected temporal graphs and the approximability of temporally connected spanning subgraphs with minimum total weight.

\smallskip\noindent{\bf Temporal Graphs and Temporal Connectivity.}
A \emph{temporal graph} is defined on a time-invariant set of $n$ vertices. Each (undirected) edge $e$ is associated with a set of discrete time labels denoting when $e$ is available. If every edge is associated with a single time label, as in \cite{KKK00}, the temporal graph is \emph{simple}. An edge $e$ available at time $t$ comprises a temporal edge $(e, t)$ and there is a positive weight $w(e, t)$ associated with it. A (resp. strict) \emph{temporal} (or \emph{time-respecting}) path is a sequence of temporal edges with non-decreasing (resp. increasing) time labels. So, temporal paths respect the time availability constraints of the edges.

Given a source vertex $r$, a temporal graph is (temporally) \emph{$r$-connected} if there is a temporal path from $r$ to any other vertex. A temporal graph is (temporally) \emph{connected} if there exists a temporal path from any vertex to any other vertex. We study the existence of dense minimally connected temporal graphs and the optimization problems of computing a minimum weight subset of temporal edges that preserve either $r$-connectivity or connectivity. We refer to these optimization problems as (Minimum) Single-Source Temporal Connectivity (or $\rTC$, in short) and (Minimum) All-Pairs Temporal Connectivity (or $\TC$, in short). They arise as natural generalizations of Minimum Spanning Tree and Minimum Arborescence in temporal networks, and to the best of our knowledge, their approximability has not been determined so far (but see \cite{AGMS15,HFL15} for some results on variants or special cases).

\smallskip\noindent{\bf Previous Work and Motivation.}
The model of simple temporal graphs with discrete time labels was introduced in \cite{KKK00}. It is essentially equivalent to the model of scheduled networks \cite{Ber96}, where each edge is available in a time interval. \cite{Ber96,KKK00} investigated how time availability restrictions on the edges affect certain graph properties. Berman \cite{Ber96} presented an algorithm for reachability by temporal paths and proved that an analogue of the max-flow-min-cut theorem holds for temporal graphs. Kempe et al. \cite{KKK00} focused on vertex-disjoint temporal paths and showed that Menger's theorem does not generalize to temporal graphs. They also identified a simple forbidden topological minor for Menger's theorem in temporal graphs. Mertzios at al. \cite{MMCS13} introduced multiple labels per edge and studied the number of temporal edges required for a temporal design to guarantee certain graph properties. Interestingly, they proved that a variant of Menger's theorem, which also takes time into account, holds in all temporal graphs. A key technical tool in \cite{Ber96,KKK00,MMCS13} is the time-expanded version of a temporal graph, which reduces reachability, edge-disjoint path and vertex-disjoint path questions in temporal graphs to similar questions in standard directed graphs.

Our motivation comes from a natural open question in \cite[Section~6]{KKK00}. Attempting an analogy between spanning trees of (standard undirected) graphs and connectivity certificates of temporal graphs, Kempe et al. asked whether any simple temporal graph admits a sparse connectivity certificate. They observed that any $r$-connected temporal graph has a time-respecting arborescence with $n-1$ edges that serves as a sparse $r$-connectivity certificate. For general temporal connectivity, however, minimum temporal connectivity certificates of different graphs have different sizes. Kempe et al. observed that an allocation of time labels to the edges of the hypercube makes it minimally temporally connected. Hence, there are temporal graphs on $n$ vertices with temporal connectivity certificates of $\Omega(n \log n)$ edges. An open question in \cite[Section~6]{KKK00} was to determine the tightest function $c(n)$ for which any temporally connected graph on $n$ vertices has a temporal connectivity certificate with at most $c(n)$ edges. A trivial upper bound on $c(n)$ is $O(n^2)$, since taking $n$ time-respecting arborescences, each rooted at a different vertex, results in a temporally connected subgraph. Kempe et al. observed that if we consider strict temporal paths and allow for the same time label at different edges, $c(n) = \Omega(n^2)$ (e.g., consider $K_n$ with the same time label on all edges). Nevertheless, for connectivity with strict temporal paths and distinct time labels, the best known lower bound on $c(n)$ is $\Omega(n \log n)$ (\cite{AGMS15}, by a labeling of the hypercube).

\smallskip\noindent{\bf Contribution.}
In this work, we resolve the open question of \cite{KKK00} and derive upper and lower bounds on the approximability of Single-Source and All-Pairs Temporal Connectivity.

In Section~\ref{s:lower}, we construct a family of simple temporal graphs with $3n$ vertices and roughly $n(n+9)/2$ edges which are almost minimally temporally connected, in the sense that the removal of any subset of $5n$ edges results in a disconnected temporal graph%
\footnote{Based on Theorem~\ref{th:lower}, we can easily obtain a family of minimally connected temporal graphs with $\Omega(n^2)$ edges (e.g., we remove temporal edges from the graph, as long as connectivity is preserved). For simplicity and clarity, we avoid presenting a tight (but more complicated) construction of dense minimally connected temporal graphs, and stick to almost minimal graphs in
the proof of Theorem~\ref{th:lower}.}
(Theorem~\ref{th:lower}). Hence, we show that $c(n) = \Theta(n^2)$  (i.e., there are graphs with dense minimum temporal connectivity certificates), thus resolving the open question of \cite{KKK00}. Our construction is essentially best possible and can be easily extended to connectivity by strict temporal paths (with distinct time labels on the edges). An interesting feature of our construction (and an indication of its tightness) is that slightly increasing the time label of a single temporal edge results in a temporal connectivity certificate with $O(n)$ edges!

Given the huge gap on the size of temporal connectivity certificates, it is natural to ask about the complexity and the approximability of Single-Source and All-Pairs Temporal Connectivity. Previous work shows that we can decide if a temporal graph is connected in polynomial time (see e.g., \cite{AGMS15,Ber96,KKK00}) and that Single-Source Temporal Connectivity can be solved in polynomial time in the unweighted case. 
Another interesting observation is that if we use the time-expanded version of a temporal graph for Minimum Temporal Connectivity, the resulting optimization problems are quite similar to Group Steiner Tree problems. In fact, this observation serves as the main intuition behind several of our results.

In Section~\ref{s:rooted}, we show that the polynomial-time approximability of Single-Source Temporal Connectivity ($\rTC$) is closely related to the approximability of the classical Directed Steiner Tree problem. Using a transformation from Directed Steiner Tree to $\rTC$ (Theorem~\ref{th:rooted_lb}) and \cite[Theorem~1.2]{HK03}, we show that $\rTC$ cannot be approximated within a ratio of $O(\log^{2-\eps} n)$, for any constant $\eps > 0$, unless $\NP \subseteq \ZTIME(n^{\poly\log n})$. Our transformation also implies that any $o(n^{\eps})$-approximation for $\rTC$ would improve the best known approximation ratio of Directed Steiner Tree. On the positive side, using a transformation from $\rTC$ to Directed Steiner Tree and the algorithm of \cite{Char98}, we obtain a polynomial-time $O(n^\eps)$-approximation, for any constant $\eps>0$, and a quasipolynomial-time $O(\log^3 n)$-approximation for $\rTC$ (Theorem~\ref{th:rooted_up}). We also show that $\rTC$ can be solved in polynomial time 
if the underlying graph has bounded treewidth (Theorem~\ref{th:rooted_twidth}).

In Section~\ref{s:mtcs}, we consider the approximability of All-Pairs Temporal Connectivity ($\TC$). Theorem~\ref{th:rooted_up} implies an $O(n^{1+\eps})$-approximation for $\TC$  (Corollary~\ref{cor:appr-rtc}). An approximation-preserving reduction to Directed Steiner Forest and \cite[Theorem~1.1]{FKN09} imply a polynomial-time $O((\Delta M)^{2/3+\eps})$-approximation for $\TC$, 
where $M$ is the number of temporal edges and $\Delta$ is the maximum degree of the underlying graph (Theorem~\ref{th:appr-dsf}). If $M$ is quasilinear and $\Delta$ is polylogarithmic, we obtain an $O(n^{2/3+\eps})$-approximation. On the negative side, a reduction from Symmetric Label Cover implies that $\TC$ cannot be approximated within a factor of $O(2^{\log^{1-\eps} n})$ unless $\NP \subseteq \DTIME(n^{\poly\log n})$ (Theorem~\ref{th:inappr}, see also \cite[Section~4]{DK99}). We also show that the unweighted version of $\TC$ is $\APX$-hard (Theorem~\ref{th:apx-hard}).

In Section~\ref{s:special}, we show that $\TC$ can be solved optimally, in polynomial time, if the underlying graph is a tree (Theorem~\ref{th:tree}), and that $\TC$ is $2$-approximable if the underlying graph is a cycle (Theorem~\ref{th:cycle}, but it is open whether $\TC$ remains $\NP$-hard for cycles).

For clarity, we focus on connectivity by (non-strict) temporal paths. However, all our results can be extended (with small changes in the proofs and with the same approximation guarantees and running times) to the case of connectivity by strict temporal paths.

\smallskip\noindent{\bf Comparison to Previous Work.}
Akrida et al. \cite{AGMS15} study connectivity by strict temporal paths. Allocating distinct time labels to the hypercube, they obtain a minimal temporally connected graph with $\Omega(n \log n)$ edges. They also show that any allocation of distinct labels to $K_n$ results in a temporal graph that is not minimally connected. However, they do not give any lower bound on the size of temporal connectivity certificates for $K_n$. Our Theorem~\ref{th:lower} 
improves on the lower bound of \cite{AGMS15} from $\Omega(n \log n)$ to $\Omega(n^2)$. \cite{AGMS15} also shows that computing the maximum number of edges that are redundant for temporal connectivity is $\APX$-hard.

Huang et al. \cite{HFL15} consider the Single-Source (but not the All-Pairs) version of Minimum Temporal Connectivity in simple scheduled networks \cite{Ber96}. They show that the problem is $\APX$-hard. Using a transformation to Directed Steiner Tree, they show that the approximation guarantees of \cite{Char98} carry over to Single-Source Temporal Connectivity for scheduled networks. Although the approximation guarantees are the same, the reduction of \cite{HFL15} is slightly different and less general than ours in Theorem~\ref{th:rooted_up} (which we discovered independently). In addition to the approximability result, we present strong inapproximability bounds for $\rTC$ and show that it is polynomially solvable for graphs with bounded treewidth.

Erlebach et al. \cite{EHK15} study the problem of computing a shortest exploration schedule of a temporal graph, i.e., a shortest strict temporal walk that visits all vertices. They prove that it is $\NP$-hard to approximate the shortest exploration schedule within a factor of $O(n^{1-\eps})$, for any $\eps > 0$, and construct temporal graphs whose exploration requires $\Theta(n^2)$ steps. Since the notion of exploration schedules is much stronger than ($r$-)connectivity, their results do not have any immediate implications for $\rTC$ and $\TC$ (e.g., the $\Theta(n^2)$-explorable graphs of \cite[Lemma~4]{EHK15} admit a temporally connected subgraph with $O(n)$ edges).

%% file: prelim.tex
\section{The Model and Preliminaries}
\label{s:prelim}

%
Throughout, we let $[k] \equiv \{1, \ldots, k\}$, for any integer $k \geq 1$.
%
%
An (edge weighted) \emph{temporal graph} $\G(V, E, L)$ with vertex set $V$, edge set $E$ and lifetime $L$ is a sequence of (undirected edge-weighted) graphs $(G_t(V, E_t, w_t))_{t \in [L]}$, where $E_t \subseteq E$ is the set of edges available at time $t$ and $w_t(e)$ (or $w(e, t)$) is the nonnegative weight of each edge $e \in E_t$. We often write $\G$ or $\G(V, E)$, for brevity. A temporal graph $\G$ is \emph{unweighted} if $w(e, t) = 1$ for all $e \in E_t$ and all $t \in [L]$.
For each edge $e \in E_t$, we say that $(e, t)$ is a \emph{temporal edge} of $\G$. For each edge $e \in E$, $L_e = \{ t \in [L]: e \in E_t \}$ denotes the set of time units (or \emph{time labels}) when $e$ is available. A temporal graph is \emph{simple} if $|L_e| = 1$ for all edges $e \in E$.
We usually let $n$ denote the number of vertices and $M = \sum_e |L_e|$ denote the number of temporal edges of $\G$. For temporal connectivity problems, we can assume wlog. that at least one temporal edge is available in each time unit, and thus, $L \leq M$.
%
%
The (static) graph $G(V, E)$ is the \emph{underlying graph} of $\G$. We say that $\G$ has some (non-temporal) graph theoretic property (e.g., is a tree, a cycle, a clique, has bounded treewidth) if the underlying graph $G$ has this property.

For a vertex set $S$, $G[S]$ (resp. $\G[S]$) is the underlying (resp. temporal) graph induced by $S$. A spanning subgraph $\G'$ of a temporal graph $\G = (G_t(V, E_t, w_t))_{t \in [L]}$ is a sequence of graphs $(G'_t(V, E'_t, w_t))_{t \in [L]}$ such that $E'_t \subseteq E_t$. The total weight of $\G'$ is $\sum_{t \in [L]} \sum_{e \in E'_t} w(e, t)$.

\smallskip\noindent{\bf Temporal Connectivity.}
A \emph{temporal} (or \emph{time-respecting}) path is an alternating sequence of vertices and temporal edges $(v_1, (e_1, t_1), v_2, (e_2, t_2), \ldots, v_k, (e_k, t_k), v_{k+1})$, such that $e_i = \{ v_i, v_{i+1} \} \in E_{t_i}$, for all $i \in [k]$, and $1 \leq t_1 \leq t_2 \leq \cdots \leq t_k$. A temporal path is \emph{strict} if $t_1 < t_2 < \cdots < t_k$. Such a temporal path is from $v_1$ to $v_{k+1}$ (or a temporal $v_1 - v_{k+1}$ path).

A temporal graph $\G$ is (temporally) \emph{$r$-connected}, for a given source $r \in V$, if there is a temporal path from $r$ to any vertex $u \in V$. A temporal graph $\G$ is (temporally) \emph{connected}, if there is a temporal path from $u$ to $v$ for any ordered pair $(u, v) \in V \times V$. If all temporal paths are strict, $\G$ is strictly connected (or strictly $r$-connected). An \emph{($r$-)connectivity certificate} of $\G$ is any spanning subgraph of $\G$ that is also ($r$-)connected.

\smallskip\noindent{\bf Single-Source and All-Pairs Temporal Connectivity.}
Given a temporal graph $\G$ and a source vertex $r$, the problem of (Minimum) \emph{Single-Source Temporal Connectivity} ($\rTC$) is to compute a temporally $r$-connected spanning subgraph of $\G$ with minimum total weight. The optimal solution to $\rTC$ is a simple temporal graph whose underlying graph is a tree (see \cite[Section~6]{KKK00} and Lemma~\ref{l:rtc_tree}).
Given a temporal graph $\G$, the problem of (Minimum) \emph{All-Pairs Temporal Connectivity} ($\TC$) is to compute a temporally connected spanning subgraph of $\G$ with minimum total weight.

\smallskip\noindent{\bf Approximation Ratio.}
An algorithm $A$ has \emph{approximation ratio} $\rho \geq 1$ for a minimization problem if for any instance $I$, the cost of $A$ on $I$ is at most $\rho$ times $I$'s optimal cost.

\smallskip\noindent{\bf Directed Steiner Tree and Forest.}
To understand the approximability of $\rTC$ and $\TC$, we use reductions from and to Directed Steiner Tree and the Directed Steiner Forest.

Given a directed edge-weighted graph $G(V, E)$ with $n$ vertices, a source $r \in V$ and a set of $k$ terminals $S \subseteq V$, the Directed Steiner Tree ($\DST$) problem asks for a subgraph of $G$ that includes a directed path from $r$ to any vertex in $S$ and has minimum total weight. The best known algorithm for $\DST$ is due to Charikar et al. \cite{Char98} and achieves an approximation ratio of $O(k^\eps)$, for any constant $\eps > 0$, in polynomial time, and of $O(\log^3 k)$ in quasipolynomial time. On the negative side, \cite[Theorem~1.2]{HK03} shows that $\DST$ cannot be approximated within a factor $O(\log^{2-\eps} n)$, for any constant $\eps > 0$, unless $\NP \subseteq \ZTIME(n^{\poly\log n})$.

Given a directed edge-weighted graph $G(V, E)$ with $n$ vertices and $m$ edges, and a collection $D \subseteq V\times V$ of $k$ ordered vertex pairs, the Directed Steiner Forest ($\DSF$) problem asks for a subgraph of $G$ that contains an $s-t$ path for each $(s,t) \in D$ and has minimum total weight. \cite{FKN09} presents a polynomial-time $O(n^\eps \min\{ n^{4/5}, m^{2/3} \})$-approximation for $\DSF$, for any constant $\eps > 0$.

%% file: lower.tex
\section{A Lower Bound on the Size of Temporal Connectivity Certificates}
\label{s:lower}

In this section, we construct an infinite family of simple temporal graphs with $\Theta(n)$ vertices and lifetime $\Theta(n)$ such that any temporal connectivity certificate has $\Omega(n^2)$ edges. Our construction is essentially best possible, since any temporal graph with $n$ vertices and lifetime $L$ admits a connectivity certificate with $O(\min\{ n^2, nL\})$ edges.

\begin{theorem}\label{th:lower}
For any even $n \geq 2$, there is a simple connected temporal graph with $3n$ vertices, $n(n+9)/2-3$ edges and lifetime at most $7n/2$, so that the removal of any subset of\, $5n$ edges results in a disconnected temporal graph.
\end{theorem}

\begin{proof}[Proof sketch.]
For any even $n$, we construct a simple connected temporal graph $\G$ with $\Theta(n)$ vertices and $\Theta(n^2)$ edges so that virtually any edge is essential for temporal connectivity. 

\smallskip\noindent\emph{The Construction.}
For any even $n$, $\G$ consists of $3n$ vertices partitioned into three sets $A = \{a_1, \ldots, a_n\}$, $H = \{ h_1, \ldots, h_n\}$ and $M = \{ m_1, \ldots, m_n\}$, with $n$ vertices each.

The underlying graph $G[A]$ is the complete graph $K_n$ and comprises the \emph{dense} part of the construction with $\Theta(n^2)$ edges. The edges of $G[A]$ are partitioned into $n/2$ edge-disjoint paths $p_1, \ldots, p_{n/2}$. Each path $p_i$ has length $n-1$ and spans all vertices in $A$ (see Figure~\ref{fig:constr}.a and Lemma~\ref{l:partition}). All edges of each path $p_i$ have time label $i$.

The vertices of $H$ comprise the \emph{intermediate} part of the construction. There are no edges with both endpoints in $H$. For every $i \in [n/2]$, one endpoint of the path $p_i$ is connected to $h_{2i-1}$ and the other endpoint is connected to $h_{2i}$. Both edges have time label $i$.

The vertices of $M$ form the \emph{interconnecting} part of the construction. For each $i \in [n/2]$, we refer to $m_{2i-1}$ (resp. $m_{2i}$) as the \emph{entry vertex} (resp. the \emph{exit vertex}) for the vertices $h_{2i-1}$ and $h_{2i}$. There are two edges connecting $m_{2i-1}$ to $h_{2i-1}$ and $h_{2i}$ with labels $n/2+2i-1$ and $n/2+2i$, respectively, and two edges connecting $m_{2i}$ to $h_{2i-1}$ and $h_{2i}$ with labels $(n/2+2i-1)\epsilon$ and $(n/2+2i)\epsilon$, respectively, for some fixed $\epsilon \in (0, 1/(4n))$.
We also connect the vertices of $M$ to each other. For every $i \in [n/2-2]$, there are edges connecting $m_{2i-1}$ to $m_{2i+2}$ and to $m_n$, and a single edge connecting $m_{n-3}$ to $m_n$. To allocate time labels to these edges, we order them in decreasing order of their endpoint with higher index, breaking ties by ordering them in increasing order of their endpoint with lower index, i.e., the order is $\{ m_1, m_n \}$, $\{ m_3, m_n \}, \ldots, \{ m_{n-3}, m_n \}$, $\{ m_{n-5}, m_{n-2} \}$, $\{ m_{n-7}, m_{n-4} \}, \ldots, \{ m_{1}, m_{4} \}$. The time label of the $k$-th edge in this order is $1-(k-1)\epsilon$.
Finally, for every $i \in [n/2]$, there are an edge with time label $\epsilon$ connecting the vertex $m_{2i-1}$ to the vertex $a_{2i-1}$ in $A$ and an edge with time label $n+1$ connecting the vertex $m_{2i}$ to the vertex $a_{2i}$ in $A$ (see also Figure~\ref{fig:constr}.b).

The total number of edges is $n(n+9)/2-3$, the number of different time labels is at most $7n/2$, and each edge has a single label (see also Section~\ref{s:app:part-a}).

\begin{figure}[!t]
	\centering%
	\begin{minipage}[t]{.42\textwidth}\centering%
			\includegraphics[width=\textwidth]{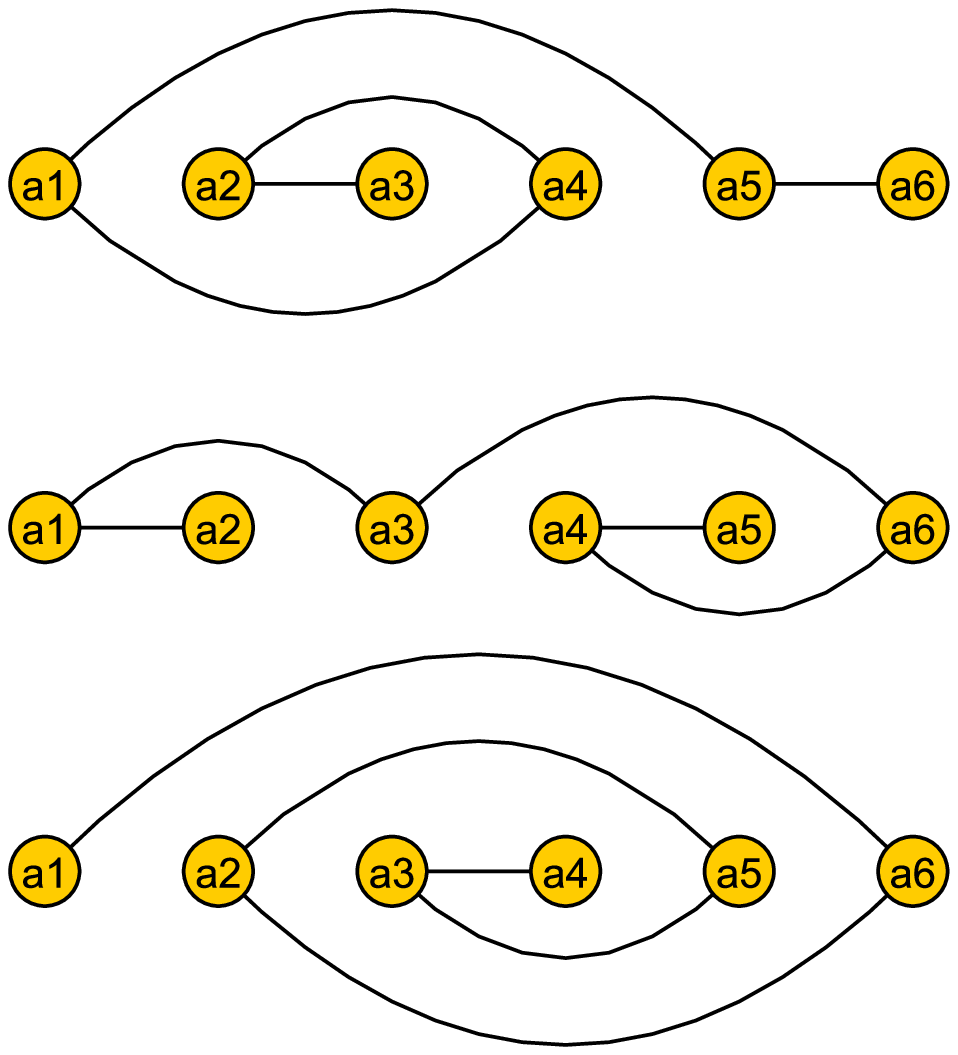}
			\caption*{(a) Partition into $3$ Hamiltonian paths.}
	\end{minipage}%
	\hfill%
	\begin{minipage}[t]{.5\textwidth}\centering%
			\includegraphics[width=\textwidth]{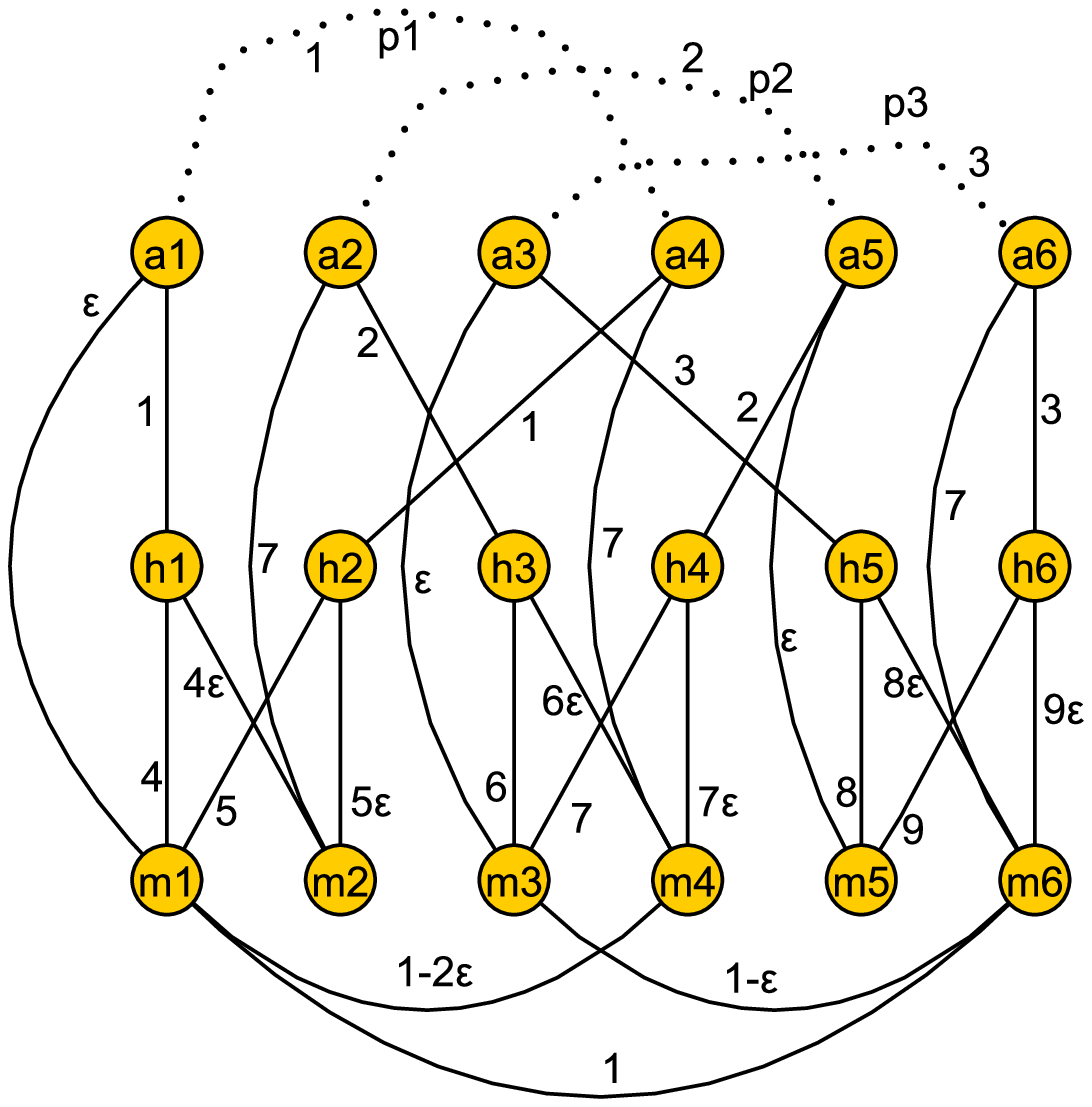}
			\caption*{\hskip15mm (b) Putting the $3$ parts together.}
	\end{minipage}
	\caption{The temporal graph constructed in the proof of Theorem~\ref{th:lower} for $n=6$.}
	\label{fig:constr}
\end{figure}

\smallskip\noindent\emph{Intuition and Main Claims.}
The construction is based on the collection $p_1, \ldots, p_{n/2}$ of $n/2$ edge-disjoint paths, where all edges in each path $p_i$ have label $i$. Extending each path $p_i$ to vertices $h_{2i-1}$ and $h_{2i}$, we get a path that connects $h_{2i}$ to $h_{2i-1}$ (and vice versa) and to all vertices in $A$ at time $i$. Moreover, different time labels make these paths essentially independent of each other, in the sense that if a temporal walk begins and ends at time $i$, it can use only edges with label $i$ (i.e., only edges of this path) to connect $h_{2i}$ to $h_{2i-1}$. In Section~\ref{s:app:part-c}, we formalize this intuition and show that the unique temporal path from $h_{2i}$ to $h_{2i-1}$ is through path $p_i$. Therefore, all edges of $G[A]$ must be present in any temporally connected spanning subgraph of $\G$.
To achieve a dense underlying graph $G[A]$, we observe that the collection of $n/2$ edge-disjoint paths can be defined so that they go through the same $n$ vertices, in a different order each (see Figure~\ref{fig:constr}.a and Lemma~\ref{l:partition}). This describes the main intuition behind our construction and explains how the dense and the intermediate parts work. The only problem now is that $H$-vertices with high indices, e.g., $h_n$, cannot reach $H$-vertices with low indices, e.g., $h_1$. The vertices in the interconnecting part $M$ serve to carefully connect each $h_j$ to each $h_i$, with $j>i+1$, without destroying the property that the only temporal path from $h_{2i}$ to $h_{2i-1}$ is through path $p_i$.

For every vertex pair $h_{2i-1}, h_{2i} \in H$, we introduce a vertex pair $m_{2i-1}, m_{2i} \in M$. As an entry vertex, $m_{2i-1}$ is connected to $h_{2i-1}$ and $h_{2i}$ with ``large'' labels (larger than $n/2$). Hence, starting from the rest of $\G$, we can reach $h_{2i-1}$ and $h_{2i}$ through $m_{2i-1}$, but we cannot continue to the edges of $p_i$ (with label $i \leq n/2$). As an exit vertex, $m_{2i}$ is connected to $h_{2i-1}$ and $h_{2i}$ with ``very small'' labels (at most $1/4$). Thus, starting from $m_{2i}$, we can reach first $h_{2i-1}$ and $h_{2i}$, and then all vertices in $A$ and any vertex $h_j$ with index $j > 2i$. Moreover, to avoid creating a temporal path from $h_{2i}$ to $h_{2i-1}$, the label of the edge $\{ h_{2i}, m_{2i-1}\}$ (resp. $\{ h_{2i}, m_{2i}\}$) is larger than the label of the edge $\{ h_{2i-1}, m_{2i-1}\}$ (resp. $\{ h_{2i-1}, m_{2i}\}$).

It remains now to connect the $M$-vertices to each other, without creating any alternative temporal paths from $h_{2i}$ to $h_{2i-1}$, for any $i \in [n/2]$. For each $i \in [n/2]$, the edges between $M$-vertices should create temporal paths from $m_{2i-1}$ and $m_{2i}$ to any vertex $m_j$ with index $j < 2i-1$. On the other hand, they should not create any temporal $m_{2i} - m_{2i-1}$ paths, since then we would have a new temporal $h_{2i} - h_{2i-1}$ path. We introduce roughly $n$ edges between $M$-vertices and carefully select their ``small'' labels in $[3/4, 1]$. Furthermore, to achieve temporal connectivity between all vertex pairs, we introduce an edge $\{ m_{2i-1}, a_{2i-1} \}$ with the minimum time label $\epsilon$ and an edge $\{ m_{2i}, a_{2i} \}$ with label $n+1$, for each $i \in [n/2]$.

To complete the proof, in Section~\ref{s:app:part-b}, we consider all possible types of ordered vertex pairs and show that the temporal graph $\G$ is indeed connected. Moreover, any subset of at least $5n$ edges includes some edges of $G[A]$. In Section~\ref{s:app:part-c}, we show that the removal of any edge from $G[A]$ with label $i$ destroys the unique temporal path from $h_{2i}$ to $h_{2i-1}$.
\end{proof}

We should highlight that increasing the label of edge $\{ a_1, m_1 \}$ from $\epsilon$ to $1$, in the graph of Theorem~\ref{th:lower}, results in a temporal graph that admits a connectivity certificate of size $\Theta(n)$ (see Lemma~\ref{l:fragile}).
Moreover, it is not hard to modify the construction of Theorem~\ref{th:lower} so that all time labels of the edges are distinct, the temporal graph $\G$ is connected by strict temporal paths, and the removal of any subset of $5n$ edges results in a disconnected temporal graph. Therefore, the quadratic lower bound of  Theorem~\ref{th:lower} also applies to connectivity by strict temporal paths and improves on the lower bound of $\Omega(n \log n)$ in \cite[Theorem~3]{AGMS15}.

%% file: rooted_mtcs.tex
\section{The Approximability of Single-Source Temporal Connectivity}
\label{s:rooted}

We proceed to study the approximability of the single-source version of Minimum Temporal Connectivity. We show that the polynomial-time approximability of $\rTC$ is closely related to the approximability of the classical Directed Steiner Tree ($\DST$) problem and that $\rTC$ can be solved in polynomial-time for graphs of bounded treewidth.

\subsection{A Lower Bound on the Approximability of $\rTC$}
\label{s:rooted_lower}

We start with an approximation-preserving transformation from $\DST$ to $\rTC$. 
The intuition is that we can use strict temporal paths to ``simulate'' the directed edges of $\DST$.

\begin{theorem}\label{th:rooted_lb}
Any polynomial-time $\rho(n)$-approximation algorithm for $\rTC$ on simple temporal graphs implies a polynomial-time $\rho(n^2)$-approximation algorithm for $\DST$.
\end{theorem}

\begin{proof}[Proof sketch.]
We present an approximation-preserving transformation from the $\DST$ to $\rTC$. Given an instance $I = (G(V,E,w), S, r)$ of $\DST$ with $|V| = n$, we construct a temporal graph $\G'$ with $n^2$ vertices so that (i) any Steiner tree connecting $r$ to $S$ in $G$ can be mapped to an $r$-connected subgraph of $\G'$ with no larger cost; and (ii) given any $r$-connected subgraph of $\G'$, we can efficiently compute a feasible Steiner tree for $I$ with no larger cost.

Each vertex $u \in V$ corresponds to a vertex $u$ in the temporal network $\G'$ of $I'$. For every directed edge $e=(u, v)$ of $G$, we create $n-1$ strict temporal $u - v$ paths of length $2$. Specifically, for every vertex $u \in V$, $\G'$ contains auxiliary vertices $z^u_i$, for all $i \in [n-1]$, and temporal edges $\{ u, z^u_i \}$ with time label $i$ and weight $0$. For every edge $e = (u, v) \in E$, $\G'$ contains temporal edges $\{ z^u_i, v\}$ with time label $i+1$ and weight $w(e)$, for all $i \in [n-1]$. Let $Z = \{ z^u_i \}_{u \in V, i \in [n-1]}$ be the set of all auxiliary vertices. For every vertex $x \in Z \cup (V \setminus S)$, $x \neq r$, $\G'$ contains a temporal edge $\{ r, x \}$ with time label $n+1$ and weight $0$. These edges ensure that $r$ is connected to all non-terminal and auxiliary vertices at no additional cost.

The temporal graph $\G'$ has $\Theta(n^2)$ vertices. Claims (i) and (ii) follow from the construction of $\G'$ and are formally proven in Section~\ref{s:app:rooted_lb}.
\end{proof}

Directed Steiner Tree cannot be approximated within a ratio of $O (\log^{2-\eps} n)$, for any constant $\eps > 0$, unless $\NP \subseteq \ZTIME(n^{\poly\log n})$ \cite[Theorem~1.2]{HK03}. Theorem~\ref{th:rooted_lb} implies that this inapproximability result carries over to $\rTC$. Moreover, any polynomial-time $o(n^{\eps})$-approximation algorithm for $\rTC$ would immediately improve the best known approximation ratio of the notoriously difficult $\DST$ problem.

\subsection{An Approximation Algorithm for $\rTC$}
\label{s:rooted_upper}

In Section~\ref{s:app:rooted_up}, we present an approximation-preserving reduction from $\rTC$ to $\DST$, thus proving Theorem~\ref{th:rooted_up} (see also the more general proof of Theorem~\ref{th:appr-dsf}). Then, we can use the algorithm of \cite{Char98} and approximate $\rTC$ within a ratio of $O(n^\eps)$, for any constant $\eps>0$, in polynomial time, and within a ratio of $O(\log^3 n)$ in quasipolynomial time. The reduction of Theorem~\ref{th:rooted_up} can be easily extended to $r$-connectivity by strict temporal paths.

\begin{theorem}\label{th:rooted_up}
Any polynomial-time $\rho(k)$-approximation algorithm for $\DST$ implies a poly\-no\-mial-time $\rho(n)$-approximation algorithm for $\rTC$ on general temporal graphs
\end{theorem}

\subsection{A Polynomial-Time Algorithm for Graphs with Bounded Treewidth}
\label{s:rooted_twidth}

In Section~\ref{s:app:rooted_twidth}, we prove that $\rTC$ can be solved in polynomial time, by dynamic programming, if the underlying graph has bounded treewidth (see e.g., \cite{DF13} about nice tree decompositions and dynamic programming algorithms for graphs of bounded treewidth).

\begin{theorem}\label{th:rooted_twidth}
Let $\G$ be a temporal graph on $n$ vertices with lifetime $L$, source vertex $r$ and treewidth at most $k$. Then, there is a dynamic programming algorithm which given a nice tree decomposition of $G$, computes an optimal solution to $\rTC$ in time $O(n k^2 3^k (L+k)^{k+1})$.
\end{theorem}

%% file: mtcs.tex
\section{The Approximability of All-Pairs Minimum Temporal Connectivity}
\label{s:mtcs}

In this section, we study the approximability of the all-pairs version of Minimum Temporal Connectivity in general temporal graphs. Reducing $\TC$ to $\rTC$ and to Directed Steiner Forest, we obtain polynomial-time approximation algorithms for $\TC$, albeit with not so strong guarantees (Corollary~\ref{cor:appr-rtc} and Theorem~\ref{th:appr-dsf}). To justify the poor approximation ratios, we reduce Symmetric Label Cover ($\SLC$) to $\TC$ and show that any $\rho(n)$-approximation for $\TC$ implies a $\rho(n^2)$-approximation for $\SLC$ (Theorem~\ref{th:inappr}). Moreover, using an approximation-preserving reduction from the Steiner Tree problem, we show that the unweighted version of $\TC$ is $\APX$-hard (Theorem~\ref{th:apx-hard}).

\subsection{Approximation Algorithms for $\TC$}
\label{s:tc-appr}

Using every vertex of the temporal graph as a source vertex and taking the union of the solutions obtained by the algorithm of Theorem~\ref{th:rooted_up} for $\rTC$, we obtain the following.

\begin{corollary}\label{cor:appr-rtc}
For any constant $\eps > 0$, there is a polynomial-time $O(n^{1+\eps})$-approximation algorithm for $\TC$ on temporal graphs with $n$ vertices.
\end{corollary}

Next, we present a reduction from $\TC$ to Directed Steiner Forest ($\DSF$) that leads to a different algorithm. Although the approximation ratio may be worse than $O(n^{1+\eps})$ in general, this algorithm gives significantly better guarantees if the total number of temporal edges is quasilinear (and if the maximum degree of the underlying graph is polylogarithmic).

\begin{theorem}\label{th:appr-dsf}
Let $\G$ be a temporal graph with $n$ vertices and $M$ temporal edges such that the underlying graph has maximum degree $\Delta$. Then, for any constant $\eps > 0$, there is a polynomial-time  $O(M^{\eps}\min\{M^{4/5},(\Delta M)^{2/3}\})$-approximation algorithm for $\TC$ on $\G$. If $M = O(n\,\poly\log n)$, we obtain an approximation ratio of $O(n^{4/5+\eps})$. If both $M = O(n\,\poly\log n)$ and $\Delta = O(\poly \log n)$, we obtain an approximation ratio of $O(n^{2/3+\eps})$.
\end{theorem}

\begin{proof}
The reduction from $\DSF$ to $\TC$ is a generalized and refined version of the reduction in Theorem~\ref{th:rooted_up}. Let $I$ be an instance of $\TC$ consisting of an underlying graph $G(V, E)$, a finite set of time labels $L_e$ for each edge $e$, and a weight $w(e, t)$ for any temporal edge $(e, t)$. We show how to transform $I$ into an instance $I'$ of $\DSF$ so that (i) any feasible solution of $I$ can be mapped to a feasible solution of $I'$ with no larger total weight; and (ii) given a feasible solution of $I'$, we can compute a feasible solution of $I$ with no larger total weight.

For convenience, we denote $H$ the edge-weighted directed graph of the $\DSF$ instance $I'$. For every temporal edge $(e,t)$ of $\G$, $H$ contains two vertices $h_{(e,t)}^1$ and $h_{(e,t)}^2$. Intuitively, $h_{(e,t)}^1$ indicates that we may use $(e, t)$ and $h_{(e,t)}^2$ indicates that we actually use $(e, t)$.
For each edge $e \in E$, let $l_1(e) < l_2(e) < ... < l_k(e)$ be the time labels in $L_e$. For every $i \in [k-1]$, $H$ contains a directed edge $(h_{(e, l_i(e))}^1, h_{(e, l_{i+1}(e))}^1)$ with weight $0$. Intuitively, these edges indicate that we can wait and use $e$ at some later time up to $l_k(e)$. Moreover, for every $i \in [k]$, $H$ contains a directed edge $(h_{(e,l_i(e))}^1, h_{(e, l_i(e))}^2)$ with weight $w(e, l_i(e))$. This edge indicates that we actually use the temporal edge $(e, l_i(e))$ and incur the corresponding cost.

For every ordered pair of temporal edges $(e_1, t_1)$, $(e_2, t_2)$ of $\G$, such that $e_1 \neq e_2$, $t_2$ is the smallest time label in $L_{e_2}$ such that $t_2 \geq t_1$ ($t_2 > t_1$, for strict  connectivity), and $e_1$ and $e_2$ share a common endpoint, $H$ contains a directed edge $(h^2_{(e_1, t_1)}, h^1_{(e_2, t_2)})$ with weight $0$.

For every vertex $v_i \in V$, $i \in [n]$, $H$ contains a pair of terminal vertices $s_i$ and $t_i$.
For every temporal edge $(e,t)$ incident to $v_i$, $H$ contains a directed edge $(s_i, h_{(e,t)}^1)$ with weight $0$ and a directed edge $(h_{(e,t)}^2, t_i)$ with weight $0$. The set of connection requirements of the $\DSF$ instance $I'$ consists of all pairs $(s_i, t_j)$ for all $i, j \in [n]$ with $i \neq j$.

By construction, any temporal $v_i - v_j$ path, which consists of a temporal edge sequence $((e_1, t_1), (e_2, t_2), \ldots, (e_k, t_k))$, with $e_1$ starting at $v_i$ and $e_k$ ending at $v_j$, corresponds to a directed $s_i - t_j$ path in $H$ of the form
\[   (s_i, h_{(e_1,t_1)}^1), (h_{(e_1,t_1)}^1, h_{(e_1,t_1)}^2), (h_{(e_1,t_1)}^2, h_{(e_2,t''_2)}^1),
   (h_{(e_2, t''_2)}^1,  h_{(e_2, t'_2)}^1), \]
\[ (h_{(e_2, t'_2)}^1,  h_{(e_2, t_2)}^1),
   (h_{(e_2,t_2)}^1, h_{(e_2,t_2)}^2), \ldots,  (h_{(e_k,t_k)}^1, h_{(e_k,t_k)}^2),
   (h_{(e_k,t_k)}^2, t_j) \]
with the same weight and vice versa. Using this observation, we can now establish claims (i) and (ii). Specifically, to show (i), we construct a feasible solution to $I'$ that includes all directed edges of weight $0$ and the directed edges $(h_{(e,t)}^1, h_{(e, t)}^2)$ corresponding to the temporal edges $(e, t)$ used in the feasible solution to $I$. Clearly, the two solutions have the same total weight and any temporal $v_i - v_j$ path in the solution to $I$ corresponds to an $s_i - t_j$ path in the solution to $I'$. To show (ii), we first observe that any directed path from some $s_i$ to some $t_j$ should include some directed edges of the form $(h_{(e,t)}^1, h_{(e,t)}^2)$ with weight $w(e, t)$. So, we construct a feasible solution to $I$ that includes the temporal edges $(e, t)$ corresponding to the positive-weight directed edges $(h_{(e,t)}^1, h_{(e,t)}^2)$ included in the feasible solution to $I'$.

In the resulting $\DSF$ instance $I'$, the total number of vertices is $O(n + M) = O(M)$ and the number of connection requirements is $O(n^2)$. If the maximum degree of the underlying graph is $\Delta$, the total number of edges is dominated by the edges of the form $(h^2_{(e_1, t_1)}, h^1_{(e_2, t_2)})$, which are $O(\Delta M)$. Applying the approximation algorithm of \cite[Theorem~1.1]{FKN09} to the $\DSF$ instance $I'$, we obtain a polynomial-time $O(M^{\eps}\min\{M^{4/5},(\Delta M)^{2/3}\})$-approximation algorithm, for any constant $\eps > 0$. In the special case where the number of temporal edges is $M = O(n \poly\log n)$, we obtain an $O(n^{4/5+\eps})$-approximation, for any constant $\eps > 0$. If both $M = O(n \poly\log n)$ and the maximum degree of the underlying graph $\Delta = O(\poly \log n)$, we obtain a polynomial-time $O(n^{2/3+\eps})$-approximation algorithm for any constant $\eps > 0$.
\end{proof}

\subsection{A Lower Bound on the Approximability of $\TC$}
\label{s:inappr}

In this section, we present an approximation-preserving reduction from Symmetric Label Cover to $\TC$. Our reduction along with standard inapproximability results for Symmetric Label Cover indicate that $\TC$ in general temporal networks is hard to approximate.

\begin{theorem}\label{th:inappr}
$\TC$ on temporal graphs with $n$ vertices cannot be approximated within a factor of $O(2^{\log^{1-\eps} n})$, for any constant $\eps > 0$, unless $\NP \subseteq \DTIME(n^{\poly\log n})$.
\end{theorem}

\begin{proof}
We present a polynomial-time approximation-preserving reduction from the Symmetric Label Cover ($\SLC$) problem to $\TC$. In $\SLC$ (see e.g., \cite[Definition~4.1]{DK99}), we are given a complete bipartite graph $H(U, W)$, with $|U| = |W|$, a finite set of colors $C$ and a binary relation $R(u, w) \subseteq C \times C$ for every vertex pair $(u, w) \in U \times W$. We seek to assign a color subset $\sigma(u) \subseteq C$ to each vertex $u \in U \cup W$ so that for every vertex pair $(u, w) \in U \times U$, there are colors $a \in \sigma(u)$ and $b \in \sigma(w)$ with $(a, b) \in R(u, w)$ and $\sum_{u \in U \cup W} |\sigma(u)|$, i.e., the total number of colors used, is minimized.

Given an instance of $\SLC$, we create a temporal graph $\G$ whose vertex set $V$ is partitioned into six sets $V_U$, $V_{C(U)}$, $V_W$, $V_{C(W)}$, $V_X$ and $\{ p, q\}$. There is a correspondence between the vertices of the bipartite graph $H$ and the vertices of $\G$ in the sets $V_U$ and $V_W$. The vertex sets $V_{C(U)} = V_U \times C$ and $V_{C(W)} = V_W \times C$ serve to encode the color assignment to the vertices of $U$ and $W$ in the $\SLC$ instance. Moreover, $V_X$ contains a vertex $(u, w, a, b)$ for every vertex pair $(u, w) \in U \times W$ and every allowable color pair $(a, b) \in R(u, w)$. Intuitively, the vertices of $V_X$ serve to ensure that the color assignment is consistent. Finally, the vertices $p$ and $q$ ensure that the temporal graph $\G$ is connected.

For every $u \in V_U$ and every $(u, a) \in V_{C(U)}$, $\G$ contains a temporal edge $\{ u, (u, a) \}$ with time label $1$ and weight $1$. Similarly, for every $w \in V_W$ and every $(w, b) \in V_{C(W)}$, $\G$ contains a temporal edge $\{ w, (w, b) \}$ with time label $4$ and weight $1$.
For every vertex $(u, w, a, b) \in V_X$, $\G$ contains the temporal edges $\{ (u, a), (u, w, a, b)  \}$ with time label $2$ and weight $0$ and $\{ (u, w, a, b), (w, b) \}$ with label $3$ and weight $0$.
The temporal graph $\G$ contains temporal edges with label $5$ and weight $0$ between $p$ and every vertex in $V_U \cup V_{C(U)} \cup V_{C(W)} \cup V_X$ and between $q$ and every vertex in $V_U$.
Moreover, $\G$ contains temporal edges with label $0$ and weight $0$ between $p$ and every vertex in $V_W$ and between $q$ and every vertex in $V_W \cup V_{C(U)} \cup V_{C(W)} \cup V_X$.
We note that the temporal graph $\G$ is connected and has $O(k^2 c^2)$ vertices, where $k = |U| = |W|$ and $c = |C|$ (in fact, the number of vertices of $\G$ is of the same order as the total size of all binary relations $R(u, w)$).

We next show that this reduction is approximation-preserving. We first show that any feasible solution to the $\SLC$ instance can be mapped to a temporally connected subgraph $\G'$ of $\G$ with at most the same weight. Let us fix any assignment $\sigma$ of a color set to each vertex of $H$ that is feasible for the $\SLC$ instance. We first include in $\G'$ all temporal edges of weight $0$. For every vertex $u \in U$ with assigned colors $\sigma(u)$, we include in $\G'$ the temporal edges $\{ u, (u, a) \}$, for all $a \in \sigma(u)$. The total weight of these edges is $|\sigma(u)|$. Similarly, for every vertex $w \in W$, we include in $\G'$ the temporal edges $\{ w, (w, b) \}$, for all $b \in \sigma(w)$. The total weight of these edges is $|\sigma(w)|$. Therefore, the total weight of the temporal subgraph $\G'$ is equal to the cost of the solution $\sigma$ for the $\SLC$ problem.

It remains to show that $\G'$ is temporally connected. All vertices in $V_U \cup V_{C(U)} \cup V_{C(W)} \cup V_X \cup \{ p \}$ are connected with each other (through $p$) by temporal edges with time label $5$. There are also temporal $p - q$ and $q - p$ paths consisting of edges with time label $5$ through the vertices of $V_U$. Similarly, all vertices in $V_W \cup V_{C(U)} \cup V_{C(W)} \cup V_X \cup \{ q \}$ are connected with each other (through $q$) by temporal edges with time label $0$. Moreover, there is a temporal path (using edges with time labels
$0$ and $5$) from every vertex in $V_W \cup V_{C(U)} \cup V_{C(W)} \cup V_X \cup \{ q \}$ to every vertex in $V_U$. Also, $p$ is connected to every vertex in $V_W$ with temporal edges of time label $0$ and vice versa. All these vertex pairs are connected through temporal paths entirely consisting of $0$-weight edges. The really interesting case concerns vertex pairs $(u, w) \in V_U \times V_W$. By the feasibility of the solution $\sigma$, for every vertex pair $(u, w) \in U \times W$, there are colors $a \in \sigma(u)$ and $b \in \sigma(w)$ such that $(a, b) \in R(u, w)$. Therefore, the temporal $u-w$ path $(u, (u, a), (u, w, a, b), (w, b), b)$ is included in $\G'$. Hence, $\G'$ is temporally connected.

We also need to show that given a temporally connected subgraph $\G'$ of $\G$, we can efficiently compute an assignment $\sigma$ of a color set to each vertex in $U \cup W$ that is feasible for the $\SLC$ instance and has total cost no larger than the total weight of $\G'$. For every $u \in V_U$ and every temporal edge of the form $(\{ u, (u, a) \}, 1)$ included in $\G'$, we include the color $a$ in $\sigma(u)$. Similarly, for every $w \in V_W$ and every temporal edge of the form $(\{ w, (w, b) \}, 4)$ included in $\G'$, we include the color $b$ in $\sigma(w)$. Since these are the only edges of $\G$ (and $\G'$) with positive weight, the total cost of $\sigma$ is equal to the total weight of $\G'$.

It remains to show that $\sigma$ is a feasible solution to the $\SLC$ instance. Let $(u,w)\in U\times W$ of the bipartite graph $H$ in the $\SLC$ instance. The crucial observation is that the only way to connect $u \in V_U$ to $w \in V_W$ in $\G'$ is through a temporal path consisting of the temporal edge sequence $(\{ u, (u, a) \}, 1)$, $(\{ (u, a), (u, w, a, b) \}, 2)$, $(\{ (w, b), (u, w, a, b) \}, 3)$, $(\{ w, (w, b)\}, 4)$, for some colors $a, b$ such that $(a, b) \in R(u, w)$. This claim immediately implies the feasibility of the assignment $\sigma$. To prove this claim, we observe that a temporal $u - w$ path cannot use any temporal edge incident to $p$ or $q$, since all edges between $V_U$ and $\{ p, q\}$ have time label $5$ and all edges between $V_W$ and $\{ p, q\}$ have time label $0$. So, any temporal $u - w$ path in $\G'$ has to move from $u$ to some vertex $(u, a) \in V_{C(U)}$. Such a vertex $(u, a) \in V_{C(U)}$ does not have any neighbors in $V_U$ other than $u$. Hence, the next vertex of any temporal $u - w$ path in $\G'$ must be to visit some $(u, w, a, b) \in V_X$, where $w \in W$ and $(a, b) \in R(u, w)$. Similarly, since such a vertex $(u, w, a, b) \in V_X$ does not have any neighbors in $V_{C(U)}$ other than $(u, a)$ and any neighbors in $V_{C(W)}$ other than $(w, b)$, we conclude that the next step of any temporal $u - w$ path in $\G'$ must be to the vertex $(w, b) \in V_{C(W)}$. But now, the last temporal edge used has time label $3$, which implies that the $u - w$ path cannot use any edges with time labels $1$ and $2$ anymore. Thus, the path cannot return to $V_U \cup V_{C(U)}$. The only choice now is that the path moves to $w$ through the temporal edge $(\{ w, (w, b)\}, 4)$, which establishes the claim about the structure of any temporal $u - v$ path in $\G'$.

The discussion above establishes the correctness of the reduction from $\SLC$ to $\TC$. Using the fact that the number of vertices of the temporal graph $\G$ is quadratic in the number of vertices of the bipartite graph $H$ and standard inapproximability results for $\SLC$ (e.g., that used in \cite{DK99}), we conclude the proof of the theorem.
\end{proof}

Adjusting the proof of Theorem~\ref{th:inappr}, we can get an approximation-preserving reduction from the {\sc MinRep} problem, which is considered in \cite{CHK11} and is similar to $\SLC$, to $\TC$. Thus, we obtain that any polynomial-time $\rho(n)$-approximation algorithm for $\TC$ on simple temporal graphs implies a polynomial-time $\rho(n^2)$-approximation algorithm for {\sc MinRep}. Since the best known approximation ratio for {\sc MinRep} is $O(n^{1/3} \log^{2/3} n)$ \cite[Section~2]{CHK11}, any $O(n^{1/6})$-approximation to $\TC$ would imply an improved approximation ratio for {\sc MinRep}.

\subsection{Inapproximability of Unweighted $\TC$}
\label{s:apx-hard}

In Section~\ref{s:app:apx-hard}, we present an approximation-preserving reduction from the Steiner Tree problem on undirected graphs with edge weights either $1$ or $2$ to $\TC$ on unweighted temporal graphs, where all temporal edges have weight equal to $1$. Since this version of the Steiner Tree problem is known to be $\APX$-hard \cite{BP89}, we obtain the following.

\begin{theorem}\label{th:apx-hard}
$\TC$ on unweighted temporal graphs is $\APX$-hard, and thus it does not admit a PTAS, unless
$\mathrm{P} = \NP$.
\end{theorem}

%% file: special.tex
\section{All-Pairs Temporal Connectivity on Trees and Cycles}
\label{s:special}

We can do better if the underlying graph is either a tree or a cycle. In Section~\ref{s:app:tree}, Lemma~\ref{l:tree}, we show that if the underlying graph is a tree, there is an optimal solution to the $\TC$ problem that uses each edge with at most two time labels. Using this structural property, we can show that $\TC$ can be solved efficiently by dynamic programming if the underlying graph is a tree (see Section~\ref{s:app:tree} for the details).

\begin{theorem}\label{th:tree}
Let $\G$ be a temporal tree on $n$ vertices with lifetime $L$. There is a dynamic programming algorithm that computes an optimal solution to $\TC$ on $\G$ in time $O(n L^{4})$.
\end{theorem}

In Section~\ref{s:app:cycle}, we observe that if the underlying graph is a cycle $C_n = (v_0, v_1, \ldots, v_{n-1}, v_0)$, any temporally connected subgraph $\G'$ can be partitioned into \emph{sectors}. A sector is a connected part $(v_i, v_{i+1}, \ldots, v_k)$ of the cycle for which there is a vertex $v_j \not\in \{ v_i, \ldots, v_{k-1} \}$ such that the temporal paths $p_{\mathrm{incr}} = (v_i, v_{i+1}, \ldots, v_j)$ and $p_{\mathrm{decr}} = (v_{k}, v_{k-1}, \ldots, v_{j+1})$ are present in $\G'$ (the vertex indices along $C_n$ are taken modulo $n$). Intuitively, any vertex in the sector $(v_i, v_{i+1}, \ldots, v_k)$ can reach every vertex in $C_n$ through the paths $p_{\mathrm{incr}}$ and $p_{\mathrm{decr}}$. Then, in Lemma~\ref{l:cycle}, we show that there is an optimal solution to the $\TC$ problem on $C_n$ where each edge is shared by at most two different sectors. Then, ignoring edges shared by different sectors and using dynamic programming to determine a near optimal partitioning of $C_n$ into sectors, we obtain the following.

\begin{theorem}\label{th:cycle}
There is a polynomial-time $2$-approximation algorithm for the $\TC$ problem on any temporal cycle $C_n$. \end{theorem}

%% file: appendix_lower.tex
\section{The Structure of Solutions to Single-Source Temporal Connectivity}
\label{s:app:rtc_tree}

The following lemma is implicit in \cite[Section~6]{KKK00} and shows that any feasible solution to $\rTC$ can be transformed to a simple temporal tree without increasing its total weight. Thus, we can always assume that the optimal solution to $\rTC$ is a simple temporal tree.

\begin{lemma}\label{l:rtc_tree}
Given a feasible solution $T'$ for $\rTC$, we can obtain a feasible solution $T$ such that (i) $T$ is a simple temporal graph, (ii) the total weight of $T$ does not exceed the total weight of $T'$, and (iii) the set of edges in $T$ form a tree in the underlying graph.
\end{lemma}

\begin{proof}
We consider any instance of $\rTC$ and any feasible solution $T'$. If $T'$ does not include a unique temporal path from $r$ to any other vertex $u$, we keep the path reaching $u$ earliest and remove all other $r-u$ paths from $T'$. So, we do not increase the total weight of $T'$ and obtain a feasible solution $T$ that satisfies (ii) and (iii).

As for (i), let $e = \{u, v\}$ be an edge in $T$. The removal of $e$ partitions the underlying tree into two connected components with vertex sets $V_1$ and $V_2$. Let $r, u \in V_1$ and $v \in V_2$. Thus, any temporal path from $r$ to a vertex in $V_2$ goes through $e$. Let $t_e$ be the earliest time that a path from $r$ goes from $u$ to $v$ through $e$. Since there is a temporal $r-v$ path that reaches $v \in V_2$ at time $t_e$, any path from $r$ to a vertex $v' \in V_2$ can use this path to reach $v$ at $t$ and then proceed to $v'$ using the temporal $v - v'$ path in $T$. The temporal $v - v'$ path must use time labels at least $t_e$, since $t_e$ is the earliest time that we reach $v$ from $r$ in $T$. Thus, we can keep the temporal edge $(e, t_e)$ and remove any other time label associated with $e$ in $T$.
\end{proof}

\section{Proof of Theorem~\ref{th:lower}: Omitted Technical Details}
\label{s:app:lower}

\subsection{Partitioning a Complete Graph into Hamiltonian Paths}
\label{s:app:partition}

We first show that for any even $n \geq 2$, the edges of the complete graph $K_n$ can be partitioned into $n/2$ Hamiltonian paths, each visiting the vertices of $K_n$ in a different order.

\begin{lemma}\label{l:partition}
For any even $n \geq 2$, the edges of the complete graph $K_n$ can be partitioned into $n/2$ Hamiltonian paths.
\end{lemma}

\begin{proof}
Let $a_0, a_1, \ldots, a_{n-1}$ be the vertices%
\footnote{For convenience, we index the vertices of $K_n$ from $0$ to $n-1$, instead of indexing them from $1$ to $n$, as in the proof of Theorem~\ref{th:lower}. After partitioning the vertices into $n/2$ Hamiltonian paths, we can increase all indices by one and become consistent with the notation in the construction of Theorem~\ref{th:lower}.}
of $K_n$. For every $i \in \{ 0, 1,..., (n-2)/2 \}$, we consider the Hamiltonian path
\( p_i = (a_i, a_{i-1}, a_{i+1}, a_{i-2}, a_{i+2}, \ldots, a_{i-n/2+1}, a_{i-n/2-1}, a_{i-n/2}) \),
where all indices are taken modulo $n$.

We claim that for any $j \neq i$, the paths $p_j$ and $p_i$ do not have any edges in common. To this end, we observe that for any $k \in [n-1]$, the absolute difference between the indices of the $k$-th vertex and $(k+1)$-th vertex in any such path is exactly $k$. Therefore, the edge set of each path $p_i$ is uniquely determined by its first vertex $a_i$. Using this observation inductively, we conclude that for any $j \neq i$, the edge sets of the paths $p_j$ and $p_i$ are disjoint.
\end{proof}

\subsection{The Number of Edges and the Lifetime of $\G$}
\label{s:app:part-a}

The total number of edges in $G$ is $n(n+9)/2-3$. Specifically, we have $n(n-1)/2$ edges between vertices in $A$, $2n$ edges connecting vertices in $H \cup M$ to vertices in $A$, $2n$ edges connecting vertices in $M$ to vertices in $H$, and $2(n/2-2)+1$ edges between $M$-vertices.

The total number of different time labels is at most $7n/2$ and each edge has a single label (so the temporal graph is simple). Specifically, we use $n/2$ labels for the edges in $G[A]$ and for the edges connecting vertices in $H$ to vertices in $A$, $2n$ labels for the edges connecting vertices in $M$ to vertices in $H$, $2$ labels for the edges connecting vertices in $M$ to vertices in $A$, and at most $n-2$ labels for the edges in $G[M]$. We also note that the use of non-integral labels in our construction is just for simplicity and without loss of generality.

\subsection{The Temporal Graph $\G$ is Connected}
\label{s:app:part-b}

We proceed to show that the temporal graph $\G$ constructed in the proof of Theorem~\ref{th:lower} is indeed connected. To this end, we consider all possible types of ordered vertex pairs $(u, v)$ and show that $\G$ contains a temporal path from $u$ to $v$. We need to distinguish between several different cases:

\begin{description}
\item[$u, v \in A$:] We move along the path $p_1$ of $\G[A]$, where all edges have label $1$.

\item[$u \in A$, $v = h_i \in H$:] We move from $u$ to $a_i$ using the path $p_1$, where all edges have label $1$. Then, we follow the edge with label $\lfloor i/2 \rfloor$ from $a_i$ to $h_i$.

\item[$u \in A$, $v = m_{2i-1} \in M$:] We first move to $h_{2i-1}$, as in the previous case, and then follow the edge with label $n/2+2i-1$ from $h_{2i-1}$ to $m_{2i-1}$.

\item[$u \in A$, $v = m_{2i} \in M$:] We move from $u$ to $a_{2i}$ using the path $p_1$, where all edges have label $1$. Then, we follow the edge with label $n+1$ from $a_{2i}$ to $m_{2i}$.

\item[$u = h_i \in H$, $v \in A$:] We move from $h_i$ to $a_i$ using the edge with label $\lceil i/2 \rceil$ and proceed to $v$ along the path $p_{n/2}$, where all edges have label $n/2$.

\item[$u = h_i \in H$, $v = h_j \in H$ and $\lceil i/2 \rceil \leq \lceil j/2 \rceil$:]
We first move from $h_i$ to $a_i$, using the edge with label $\lceil i/2 \rceil$, and then on the path $p_{\lceil i/2 \rceil}$ to $a_j$. Then, we move to $h_j$ using the edge with label
$\lceil j/2 \rceil$.

\item [$u = h_i \in H$, $v = h_j \in H$ and $\lceil i/2 \rceil > \lceil j/2 \rceil$:]
We let $\ell_i = 2 \lceil i /2 \rceil$ and $\ell_j = 2 \lceil j /2 \rceil$. We first move from $h_i$ to $m_{\ell_i}$ using the edge with ``very small'' label $(n/2+i)\epsilon$. We then move to $m_{\ell_i-3}$, to $m_n$ and to $m_{\ell_j-1}$ using the edges between $M$-vertices with ``small'' labels. Finally, we move from $m_{\ell_j-1}$ to $h_j$ using the edge with ``large'' label $n/2+j$. We note that the labels of the temporal edges between $M$-vertices are increasing due to the particular ordering that we use.

\item[$u = h_i \in H$, $v = m_{2j-1} \in M$:]
If $\lceil i /2 \rceil \leq j$, we move from $h_i$ to $h_{2j}$ as in the corresponding case, using edges with labels at most $n/2$. Otherwise, we move from $h_i$ to $h_{2j}$ as in the previous case, using edges with labels at most $n/2+2j$. Finally, we move from $h_{2j}$ to $m_{2j-1}$ using the edge with label $n/2+2j$.

\item[$u = h_i \in H$, $v = m_{2j} \in M$:]
We move from $h_i$ to $a_{2j}$ as in the corresponding case, using labels at most $n/2$, and from $a_{2j}$ to $m_{2j}$ using the edge with label $n+1$.

\item[$u = m_{2i-1} \in M$ to any other vertex $v$:]
We move from $m_{2i-1}$ to $a_{2i-1}$ using the edge with $\epsilon$ and from $a_{2i-1}$ to any other vertex, as in the first four cases, since $\epsilon$ is the minimum time label.

\item[$u = m_{2i} \in M$, $v \in A$:]
We move from $m_{2i}$ to $h_{2i}$, using the edge with label $(n/2+2i)\epsilon$, next to the second endpoint of path $p_i$, using the edge with label $i$, and next to any vertex in $A$, along the path $p_i$, where all edges have label $i$.

\item[$u = m_{2i} \in M$, $v = m_{2j} \in M$:]
We move from $m_{2i}$ to $a_{2j}$ as in the previous case, using labels at most $j$, and next to $m_{2j}$, using the edge with label $n+1$.

\item[$u = m_{2i} \in M$, $v = h_j \in H$:]
If $i \leq \lceil j /2 \rceil$, we move from $m_{2i}$ to $h_{2i}$ using the edge with label $(n/2+2i)\epsilon$, next to the second endpoint of path $p_i$, using the edge with label $i$, next to vertex $a_j$ along the path $p_i$, where where all edges have label $i$, and finally to $h_j$ using the edge with label $\lceil j /2 \rceil$.
Otherwise, we move from $m_{2i}$ to $m_{2i-3}$, to $m_n$ and to $m_{2\lceil j /2 \rceil-1}$, using the edges between $M$-vertices, with ``small'' labels, and finally from $m_{2\lceil j /2 \rceil-1}$ to $h_j$, using the edge with label $n/2 + j$.

\item[$u = m_{2i} \in M$, $v = m_{2j-1} \in M$:]
If $i \leq j$, we move to $h_{2j-1}$, as in the first walk of the previous case, using labels at most $j$, and next to $m_{2j-1}$ using the edge with label $n/2+2j-1$. Otherwise, we follow the second walk of the previous case up to $m_{2j-1}$. \qed
\end{description}

\subsection{Removing a Linear Number of Edges Disconnects $\G$}
\label{s:app:part-c}

Since the total number of edges in $\G$ is $n(n+9)/2-3$ and the dense part $\G[A]$ includes $n(n-1)/2$ of them, removing any subset of at least $5n$ edges results in the removal of at least one edge $e$ between $A$-vertices. Let us assume that $e$ has label $i$. We next show that the removal of any edge with label $i$ results in a temporal graph $\G'$ with no temporal path from $h_{2i}$ to $h_{2i-1}$. Hence, the removal of any set of at least $5n$ edges results in a disconnected temporal graph.

To reach a contradiction, let us assume that there is a temporal path $p$ in $\G'$ from $h_{2i}$ to $h_{2i-1}$. Since an edge of label $i$ has been removed, the path $p_i$ is not present in $\G'$. Therefore, $p$ should use some edges with labels other than $i$. We distinguish between two cases:

\begin{description}
\item[The first edge of $p$ has label $i$ and the last edge has label larger than $i$.]
The only edge incident to $h_{2i-1}$ that can be used as the last edge of $p$ has label $n/2+2i-1$. We note that any edges with labels at most $1$ cannot be used by $p$. The same is true for the edges with label $n+1$, which connect vertices $m_{2j}$ and vertices $a_{2j}$, for every $j \in [n/2]$, since all other edges incident to vertices $m_{2j}$ have labels at most $1$ and all other edges incident to $a_{2j}$ have labels at most $n/2$. If we ignore all edges with labels at most $1$ and $n+1$, we have no edges between $A$-vertices and $M$-vertices and no edges between $M$-vertices. Therefore, $p$ has to visit vertex $a_{2i}$ first and then to move through the vertices of $A$. At some point, a vertex $m_j \in M$ must be visited, before $p$ finally reaches $h_{2i-1}$. But this leads to a contradiction, since there are no usable edges between $A$-vertices and $M$-vertices. Hence, no such path $p$ from $h_{2i}$ to $h_{2i-1}$ exists.

\item[Otherwise.]
The only alternative for $p$ is to begin with label $(n/2 + 2i)\epsilon$ and to end with either label $i$ or label $n/2+2i-1$. Then, the first edge brings $p$ to $m_{2i}$. At that point, any edge with the minimum label $\epsilon$ cannot be used by $p$. As before, the same is true for the edges with label $n+1$, which connect vertices $m_{2j}$ and vertices $a_{2j}$, for every $j \in [n/2]$, since all other edges incident to vertices $m_{2j}$ have labels at most $1$ and all other edges incident to $a_{2j}$ have labels at most $n/2$. So, we ignore all edges with labels either $\eps$ or $n+1$. Then, $p$ can only visit vertices $m_{2j-1}$, with $j < i$, or the vertex $m_n$ (through some edge with a `small'' label in $[3/4, 1]$). From $m_n$, $p$ cannot proceed any further, since all edges that connect $m_n$ to some vertex in $H$ have ``very small'' labels in $(0, 1/4]$. Otherwise, if the last $M$-vertex visited by $p$ is $m_{2j-1}$, for some $j < i$, the next edge of $p$ must have either label $n/2+2j-1$ (vertex $h_{2j-1}$) or label $n/2+2j$ (vertex $h_{2j}$), and then, $p$ cannot proceed any further. Again, we reach a contradiction. Hence, no such path $p$ from $h_{2i}$ to $h_{2i-1}$ exists. \qed
\end{description}

\subsection{Tightness of the Construction: Increasing a Single Label}
\label{s:app:part-d}

We next show that if in the temporal graph $\G$ constructed in Theorem~\ref{th:lower}, we increase the label of edge $\{ a_1, m_1 \}$ from $\epsilon$ to $1$, the resulting temporal graph $\G'$ admits a connectivity certificate with $\Theta(n)$ edges.

\begin{lemma}\label{l:fragile}
Let $\G$ be the temporal network constructed in Theorem~\ref{th:lower} and let $\G'$ be the temporal network obtained from $\G$ if we increase the label of edge $\{ a_1, m_1 \}$ from $\epsilon$ to $1$. Then, $\G'$ admits a temporal connectivity certificate with $\Theta(n)$ edges.
\end{lemma}

\begin{proof}
We show that the temporal graph $\G'$ remains connected if we remove the edges of all paths $p_i$ with $i > 1$. Clearly, this results in a temporal graph with $\Theta(n)$ edges. The proof shows how to modify all temporal $u - v$ paths that are described in Section~\ref{s:app:part-b} and use some edges of a path $p_i$, with $i > 1$. Again, we need to consider several cases for the initial vertex $u$ and the final vertex $v$:

\begin{description}
\item[$u = h_i \in H$, $v \in A$:]
We move from $h_{i}$ to $m_{2i}$, to $m_{2i-3}$, to $m_n$, and to $m_1$, using only edges with labels at most $1$. Using the edge $\{ m_1, a_1 \}$, with label $1$, we can visit any vertex in $A$ through the path $p_1$, where all edges have label $1$.

%

\item[$u = h_i \in H$, $v = h_j \in H$:]
As in the first case, we move from $h_{i}$ to the endpoint of the path $p_{\lceil j/2 \rceil}$ which is connected to $h_j$. This part uses labels at most $1$. For the last step, we use the edge from the endpoint of $p_{\lceil j/2 \rceil}$ to $h_{j}$, which has label $\lceil j/2 \rceil \geq 1$.

\item[$u = m_{2i} \in M$, $v \in A$:]
Similarly to the first case, we move from $m_{2i}$ to $m_{2i-3}$, to $m_n$, and to $m_1$, using only edges with labels at most $1$. Using the edge $\{ m_1, a_1 \}$, with label $1$, we can visit any vertex in $A$ through the path $p_1$, where all edges have label $1$.

\item[$u=m_{2i} \in M$, $v=h_j \in H$:]
As in the previous case, we from $m_{2i}$ to the endpoint of the path $p_{\lceil j/2 \rceil}$ which is connected to $h_j$. This part uses labels at most $1$. For the last step, we use the edge from the endpoint of $p_{\lceil j/2 \rceil}$ to $h_{j}$, which has label $\lceil j/2 \rceil \geq 1$.
\end{description}
\end{proof} 

%% file: appendix_rooted.tex
\section{Single-Source Temporal Connectivity}
\label{s:app:alg}

\subsection{Proof of Theorem~\ref{th:rooted_lb}: Omitted Technical Details}
\label{s:app:rooted_lb}

We next prove that (i) any feasible solution $T$ of the $\DST$ instance $I = (G(V,E,w), S, r)$ can be mapped to a temporally $r$-connected subgraph $T'$ of $\G'$ with no larger cost; and (ii) given a temporally $r$-connected subgraph $T'$ of $\G'$, we can efficiently compute a directed Steiner tree $T$ which connects $r$ to all terminal vertices in $S$ and has no larger cost than $T'$.

We first show (i). Let $T$ be any feasible solution to $\DST$. For every vertex $u$, let $d(u)$ be the number of edges on the $r - u$ path in $T$ (we consider the $r-u$ path with the smallest number of edges, if there are many). To transform $T$ into an $r$-connected subgraph $T'$ of $\G'$, we first include in $T'$ all $0$-cost edges incident to $r$. Next, for each directed edge $(u,v)$ in $T$, we include in $T'$ the temporal path $u - z^u_{d(u)} - v$, with time labels $d(u)$ and $d(u)+1$ and weight $w(e)$. Clearly, the total weight of $T'$ does not exceed the total weight of $T$. Moreover, there is a temporal path from $r$ to any vertex $u$ of $\G'$. If $u \in Z \cup (V \setminus S)$, we use the direct edge $\{r, u\}$ with time label $n+1$. If $u \in S$, there is a directed $r - u$ path in $T$. For the $i$-th edge of this path, $i \leq n-1$, we have added a temporal path of length $2$ with labels $i$ and $i+1$. The union of these paths defines a temporal path from $r$ to $u$.

We next show (ii). Let $T'$ be any $r$-connected subgraph of $\G'$. Without increasing the cost of $T'$, we can assume that any edge $\{ z^u_i, v \}$ with positive weight $w(u, v)$ in $T'$ belongs to a temporal $u - v$ path of length $2$, where $u, v \in V$. To transform $T'$ into a feasible solution $T$ of the $\DST$ instance $I$, we include in $T$ the edge $(u, v) \in E$ of any temporal path $u - z^u_i - v$ in $T'$. Clearly, the total weight of $T$ does not exceed the total weight of $T'$. To show that $T$ is a feasible solution, we observe that all edges incident to a terminal vertex $u \in S$ have time labels at most $n$. Therefore, the maximum label in any temporal $r - u$ path is at most $n$. Hence, for any $u \in S$, $T'$ includes a temporal $r-u$ path consisting of temporal paths of length $2$, i.e., of paths of the form $v_k - z^{v_k}_{i} - v_{k+1}$. Each of these paths becomes a directed edge $(v_k, v_{k+1})$ in $T$. These directed edges together form a directed $r-u$ path in $T$. \qed

\subsection{The Proof of Theorem~\ref{th:rooted_up}}
\label{s:app:rooted_up}

Let $I$ be an instance of $\rTC$ consisting of an underlying graph $G(V, E)$, a source $r \in V$, a finite set of time labels $L_e$ for each edge $e$, and a cost $w(e, t)$ for any temporal edge $(e, t)$. We show how to transform $I$ into an instance $I' = (H, S, r')$ of $\DST$ so that (i) any feasible solution of $I$ can be mapped to a feasible solution of $I'$ with no larger cost; and (ii) given a feasible solution of $I'$, we can compute a feasible solution of $I$ with no larger cost.

The directed graph $H$ of the $\DST$ instance contains a non-terminal vertex for each temporal edge $(e, t)$ of $\G$. For every (ordered) pair of temporal edges $(e_1, t_1)$, $(e_2, t_2)$ of $\G$, such that $e_1 \neq e_2$, $t_1 \leq t_2$ (or $t_1 < t_2$, for strict $r$-connectivity), and $e_1$ and $e_2$ share a common endpoint, there is a directed edge $((e_1, t_1), (e_2, t_2))$ with weight $w(e_2, t_2)$ in $H$.
%
%
$H$ also contains a root vertex $r'$ corresponding to the source $r$. For every temporal edge $(e, t)$ incident to $r$ in $\G$, there is a directed edge $(r', (e, t))$ with weight $w(e, t)$ in $H$. For every vertex $u \in V \setminus\{ r\}$, $H$ contains a terminal vertex $u'$. The terminal set $S$ of $I'$ consists of all these terminal vertices $u'$. Moreover, for every temporal edge $(e, t)$ incident to $u$ in $\G$, $H$ contains a directed edge $((e, t),u')$ with weight $0$. The directed graph $H$ has $O(M)$ vertices, $n-1$ of which are terminal vertices.

We first show (i). Let $T$ be any feasible solution of the $\rTC$ instance $I$. By  Lemma~\ref{l:rtc_tree}, we can assume that $T$ is a simple temporal graph and that its underlying graph is a tree (otherwise, we can obtain such a solution from $T$ without increasing its total weight). Thus, $T$ contains a unique temporal path from $r$ to any vertex $u \in V \setminus \{ r \}$. To obtain a feasible solution $T'$ for the $\DST$ instance $I'$, we do the following, starting from $r$, in a BFS order: For every pair of vertices $u, v \in V$ such that $u$ is the immediate predecessor of $v$ on the unique temporal $r - v$ path in $T$, $u$ is reached through the temporal edge $(e_1, t_1)$, and $v$ is reached from $u$ through the temporal edge $(e_2, t_2)$, we add the directed edge $((e_1, t_1), (e_2, t_2))$ to $T'$ (this edge exists in $H$, because both $e_1$ and $e_2$ have $u$ as one of their endpoints and $t_1 \leq t_2$). For every vertex $v \in V$ reachable in $T$ directly from $r$ through the temporal edge $(e, t)$, we add the directed edge $(r', (e, t))$ to $T'$. Moreover, we add all directed edges of weight $0$ to $T'$. Clearly, the cost of $T'$ does not exceed the cost of $T$. As for the feasibility of $T'$, we observe that any temporal $r -u$ path in $T$, which consists of a temporal edge sequence $((e_1, t_1), (e_2, t_2), \ldots, (e_k, t_k))$, with $e_1$ starting at $r$ and $e_k$ ending at $u$, corresponds to a directed $r' - u'$ path $(r', (e_1, t_1), (e_2, t_2), \ldots, (e_k, t_k), u')$ in $T'$

We next show (ii). Let $T'$ be any feasible solution of the $\DST$ instance $I'$. Thus, $T'$ includes a directed path from $r'$ to any vertex $u' \in S$. For each Steiner vertex $(e, t)$ used by some path in $T'$, we add the corresponding temporal edge to the solution $T$ of the $\rTC$ instance $I$. Clearly, the total cost of $T$ is no more than the total weight of $T'$. We also need to show that $T$ is indeed a feasible solution for the $\rTC$ instance $I$. For every vertex $u \in V \setminus \{ r \}$, there is a directed path $(r', (e_1, t_1), (e_2, t_2), \ldots, (e_k, t_k), u')$ in $T'$. By the construction of $I'$, $t_1 \leq t_2 \leq \cdots \leq t_k$ and $(e_1, e_2, \ldots, e_k)$ form an $r - u$ path in the underlying graph of $\G$. Therefore, since the temporal edges $(e_1, t_1), (e_1, t_2), \ldots, (e_k, t_k)$ are present in $T$, $T$ includes a temporal $r - u$ path for every vertex $u \in V \setminus \{ r \}$.
\qed

\def\time{\mathrm{time}}

\subsection{The Proof of Theorem~\ref{th:rooted_twidth}}
\label{s:app:rooted_twidth}

In this section, we present an exact polynomial-time algorithm for $\rTC$ on temporal graphs with treewidth at most $k$. The algorithm is based on dynamic programming and its time complexity is $O(n k^2 3^k (L+k)^{k+1})$.
To this end, we consider a temporal graph $\G(V, E)$ on $n$ vertices with lifetime $L$, source vertex let $r$ and treewidth $k$. Below, we use the standard notation and terminology employed in dynamic programming algorithms for graph theoretic problems on graphs of bounded treewidth (see e.g., \cite[Chapter~10]{DF13}).

We assume a \emph{nice tree decomposition} of $\G$ with $O(n k)$ bags, where the root bag consists only of the vertex $r$. Also, for each bag $i$, consider an enumeration $v_1^i, v_2^i, \ldots, v_p^i$ of the vertices in bag $i$, where $p\leq k+1$. For each bag $i$, we define $f(i, a_1, t_1, \ldots, a_{k+1}, t_{k+1})$, where $a_j \in \{0,1\}$, $t_j\in [L+k+1]$, as follows:

\begin{itemize}
\item Let $H$ be the subgraph of $\G$ induced on the vertices of bag $i$ and all of its descendants in the tree decomposition.

\item Let $T$ be a minimum cost temporal subgraph of $H$ such that for every vertex $v$ of $H$, there exists a vertex $v_j^i$ with $a_j=1$ and a temporal path from $v_j^i$ to $v$ in $T$, starting no sooner than $t_j$.

\item Moreover, if $v=v_l^i$ for some $l$, then this temporal path ends at time $t_l$.
\end{itemize}

Note that $H$ is always a temporal forest. Intuitively, $a_j=1$ means that the $j$-th vertex of the bag is currently connected to the root, while $a_j=0$ means that it is not connected yet, and we should connect it via some temporal path in $H$. Note that some of the arguments of $f$ can be undefined, as some bags may have less than $k+1$ vertices. The cost of the optimal solution to $\rTC$ is then $f(1,1,1)$, since we start at bag $1$ (which only contains the source $r$) and we are at $r$ at time $1$. To break the symmetry introduced by same-label components, for each set of vertices in bag $i$ with the same time label $t$, we introduce an \textit{ordering} between these vertices. This is done to avoid cyclic dependencies between parent links, which would lead to an infeasible solution. To express such an ordering $x_1,\ldots ,x_k$, suppose that the vertices $v_{x_1}^i, \ldots, v_{x_k}^i$ all share the same time label $t$ (the temporal path from $r$ to each of them ends at time $t$). Then, we set $t_{x_1}=t, t_{x_2} = L + x_1, \ldots, t_{x_k}=L+x_{k-1}$. Thus, we express the required information with space $(L+k)^{k+1}$.

Below, we let $c_1(i)$ and $c_2(i)$ denote the left and right children of bag $i$ respectively. The recurrence relation for computing $f$ is the following:
\begin{description}
\item[Case 1: $i$ is an introduce vertex.] Suppose wlog that $v_p^i$ is introduced. If $a_p=0$, then there must exist a temporal path from some vertex $v_j^i$ with $a_j=1$ to $v_p^i$ in $T$, the second-to-last vertex of the path being in bag $i$. This is because $v_p^i$ exists in bag $i$ but none of its descendants, so this is the last chance to connect it to vertices with $a_j=1$. It is also the last chance to connect vertices with $a_j=0$ directly to $v_p^i$ (as its children). To this end, we try all possible vertices as parents of $v_p^i$, and also all possible subsets of vertices with $a_j=0$ as the set of children of $v_p^i$. Of course, we adhere to the time labels of the vertices, as well as the ordering introduced above. Thus, we have that
\begin{align*}
 f(i,a_1,t_1,\ldots,0,t_p) =
   \min_{j\neq i}&\!\left\{ f(c_1(i),a_1,t_1,\ldots,1,t_p) +
                w_{\time(p)}(v_j^i,v_p^i)\right.\,:\\
      & \ \ \ \time(j)<\time(p)\,\lor\\
      & \ \ \ (\time(j)=\time(p)\land \mbox{$j$ is before $p$ in the ordering}) \Bigg\}
\end{align*}
where $\time(j)$ is the time label that the temporal path from $r$ to $v_j^i$ ends. (If $t_j\leq L$, then $\time(j)=t_j$, otherwise $\time(j) = \time(t_j-L)$, where $t_j-L$ is the index of the predecessor of $v_j^i$ in the ordering of vertices with $\time(k) = \time(j)$).
Similarly,
\begin{align*}
 f(i,a_1,t_1,\ldots,1,t_p) =
   \min_S&\!\left\{ f(c_1(i), a_1^\ast, t_1, \ldots, a_{p-1}^\ast, t_{p-1}) +
            \sum_{v_j^i\in S} w_{\time(j)}(v_p^i,v_j^i)\right.\,:\\
      & \ \ \ S \subseteq \{v_j^i\in\{v_1^i,\ldots,v_{p-1}^i\} : a_j=0\,\land\\
      & \ \ \ (\time(p)<\time(j)\,\lor\\
      & \ \ \ (\time(p)=\time(j)\land \mbox{$p$ is before $j$ in the ordering})) \}\Bigg\}
\end{align*}
where $a_j^\ast=0$ if $a_j=0$ and $v_j^i \not\in S$, and $a_j^\ast=1$ otherwise.
Therefore, the complexity of computing all values of $f$ that fall to this case is $O(n k^2 3^{k} (L+k)^{k+1})$. This breaks down into $O(nk)$ for the size of the tree, $O(k)$ for iterating over vertices of $S$ and computing $\time(j)$, $O(3^{k+1})$ for choices over vertices with $a_j=1$ as well as vertices with $a_j=0$ and $v_j^i\in S$, and $(L+k)^{k+1}$ for choices of $t_j$.

\item[Case 2: $i$ is a forget vertex.] Suppose wlog that $v_{p+1}^{c_1 (i)}$ is forgotten, and iterate over all possible values $t_{p+1}$. $v_{p+1}^i$ is either assigned a new time label, or is inserted in the ordering between vertices with the same time label. We have that
%
\begin{align*}
 f(i,a_1,t_1,\ldots,a_p,t_p) =
   \min_{t_{p+1}\in [L+k]}&\Big\{ f(c_1 (i),a_1,t_1,\ldots,a_{j'},L+p+1,\ldots,a_p,t_p,0,t_{p+1}) \,:\\
      & \ \ \ t_{p+1}\leq L\lor t_{p+1}-L=j\in [p]\,\Big\}
\end{align*}
where $v_{j'}^i$ was the successor of $v_j^i$ (if it had one) in the ordering, which now becomes the successor of $v_{p+1}^i$ in the case that $t_{p+1}-L=j$ ($v_{p+1}^i$ is inserted between
$v_j^i$ and $v_{j'}^i$ in the ordering). So, in these cases, the total time complexity is $O(n k^2  2^{k} (L+k)^{k+1})$.

\item[Case 3: $i$ is a join vertex.] Then some vertices with $a_j=0$ will be connected in $c_1(i)$ and some others in $c_2(i)$. Suppose that $P=\{ v_j^i \in \{v_1^i,\ldots,v_p^i\} : a_j=0 \}$. Therefore, trying every possible partition of the vertices with $a_j=0$, we have that
\begin{align*}
  f(i,a_1,t_1,\ldots,a_p,t_p) = \min_{S_1,S_2}& \Big\{ f(c_1(i),a_1^\ast,t_1,\ldots,a_p^\ast,t_p)
             + f(c_2(i),a_1^{\ast\ast},t_1,\ldots,a_p^{\ast\ast},t_p) \,: \\
        &\ \ S_1\cup S_2 = P \mbox{ and } S_1\cap S_2 = \emptyset \,\Big\}
\end{align*}
Here we have that $a_j^\ast=0$ and $a_j^{\ast\ast}=1$, if $a_j=0$ and $v_j^i\in S_2$, that $a_j^\ast=1$ and $a_j^{\ast\ast}=0$, if $a_j=0$ and $v_j^i\in S_1$, and that $a_j^\ast=a_j^{\ast\ast}=1$, if $a_j=1$.

Here, the total time complexity is $O(n k^2 3^k (L+k)^{k+1})$. We observe that the vertex order is introduced in the recurrence due to Case~3, since without such an ordering, it would not be possible to solve the two subproblems independently. 
\end{description}

Overall we can compute $f$, and the corresponding minimum cost $r$-connected temporal subgraph, in time $O(n k^2 3^k (L+k)^{k+1})$. \qed

%% file: appendix_mtcs.tex
\def\ST{\mathrm{ST}}
\def\OPT{\mathrm{OPT}}

\section{APX-hardness of Unweighted $\TC$: The Proof of Theorem~\ref{th:apx-hard}}
\label{s:app:apx-hard}

In this section, we show that $\TC$ on unweighted temporal graphs is $\APX$-hard and does not admit a PTAS, unless $\mathrm{P} = \NP$, using an approximation-preserving reduction from Steiner Tree on undirected graphs with edge weights $1$ and $2$ to $\TC$. 

Given an undirected edge-weighted graph $G(V, E)$ and a set of terminals $S \subseteq V$, the Steiner Tree $(\ST)$ problem asks for a connected subgraph of $G$ that spans $S$ and has minimum total weight. It is easy to see that the optimal solution is always a tree. If the weight of each edge is either $1$ or $2$, we get the Steiner Tree problem with weights $1$ and $2$ ($\ST(1, 2)$).

In \cite{BP89}, $\ST(1, 2)$ is shown $\APX$-hard by a reduction from Vertex Cover on graphs with bounded maximum degree. If $n$ and $m$ are the number of vertices and edges of the Vertex Cover instance respectively, then the $\ST(1, 2)$ instance has $m$ terminals and $n$ non-terminals. Furthermore, all the terminals have degree $2$ and the graph is bipartite, the one part consists of the terminal vertices and the other part consists of the non-terminal vertices. This means that the total number of edges is at most $2m$. In order to use uniform weights for the edges of the temporal graph, we substitute each edge of weight $2$ with a $2$-edge path, where each edge has weight $1$ and the newly added vertex is a non-terminal vertex. Since at most one non-terminal vertex per edge is added, the resulting number of non-terminal vertices is at most $n + 2m \leq 3m$, assuming that the initial graph is connected and not a tree. So we obtain an instance of the Steiner Tree problem with uniform edge-weights, a set $S$ of $m$ terminal vertices and a set $T$ of at most $3m$ non-terminal vertices.

To reduce this special case of the Steiner Tree problem to $\TC$ with uniform temporal edge weights, we create a temporal variant of the above graph, where each existing edge has label $3$. To ensure temporal connectivity, we add some new elements to our temporal graph. We add new vertices $p$ and $q$ and connect every non-terminal vertex in $T$ to $p$ with time label $4$ and to $q$ with time label $2$. These edges ensure that there is a temporal path from every vertex in $S \cup T$ (terminal and non-terminal vertices, respectively) to every other vertex in the same set. 

It remains to ensure temporal connectivity between vertices $p$ and $q$ and to the rest of the temporal graph. Suppose we enumerate the elements of $S$ as $u_1, u_2,\ldots,u_m$. We add $m$ new vertices $a_1,a_2,\ldots,a_m$, composing vertex set $A$. To create a path from $p$ to any vertex of $S$, add temporal edges with label $1$ from $p$ to every vertex in $A$, and for each $i\in [m]$, add a
temporal edge between $a_i$ and $u_i$ with label $5$. Similarly, to connect $S$ to $q$, add a new vertex set $B$ consisting of vertices $b_1,b_2,\ldots,b_m$, for every $i$ add a temporal edge between $u_i$ and $b_i$ with label $1$, and add an edge between $q$ and every vertex in $B$ with label $5$. Finally, to create a temporal path from $p$ to $q$, add a new vertex $x$, which is connected to $p$ with label $1$ and to $q$ with label $5$. Then, we solve $\TC$ on the resulting temporal graph $\G$.

We are ready to show the following. Note that Theorem~\ref{th:apx-hard} follows from Lemma~\ref{l:apx-hard} below and the fact that Steiner Tree with edge weights $1$ and $2$ is $\APX$-hard. 
\begin{lemma}\label{l:apx-hard}
An $(1+\eps)$-approximate solution to $\TC$ implies an $(1+12\eps)$-approximate solution to the Steiner Tree instance with uniform weights.
\end{lemma}

\begin{proof}
The proof consists of two main claims. The first claim is that any feasible solution to $\TC$ on the resulting temporal graph $\G$ uses all the edges that are not in the Steiner Tree instance and that none of them can be part of some temporal path between vertices of $S$. The second claim is that taking all the edges not in the Steiner Tree instance together with the set of temporal edges corresponding to a feasible solution to the Steiner Tree instance results in a temporally connected subgraph of $\G$. 

Having established these claims and assuming that the number of the additional (not in the Steiner Tree instance) edges is $k \leq 4m + 6m + 2 \leq 11m \leq 11\,\OPT$, where $\OPT$ is the optimal
solution to the Steiner Tree instance, we obtain that the $\TC$ solution has cost at most $(1+\eps)(\OPT + k)$. Deleting these $k$ edges, we have a feasible solution to Steiner Tree of cost at most $(1+\eps)\OPT + \eps k \leq (1+12\eps)\OPT$.

Let us first show that any feasible solution to $\TC$ on the temporal graph $\G$ resulting by the reduction above uses all the edges that are not in the Steiner Tree instance. We examine several cases:
\begin{description}
\item[Edge between $u_i$ and $a_i$ is missing for some $i$:]
Then there is no temporal path from $u_i$ to $a_i$, since $a_i$ is only adjacent to a temporal edge with label $1$, and there is no $1$-labeled path from $u_i$ to $a_i$.

\item[Edge between $u_i$ and $b_i$ is missing for some $i$:]
Then there is no temporal path from $b_i$ to $u_i$, since $b_i$ is only adjacent to a temporal edge with label $5$, and there is no $5$-labeled path from $b_i$ to $u_i$.

\item[Edge between $a_i$ and $p$ is missing for some $i$:]
Then there is no temporal path from $a_i$ to $p$, since $a_i$ is only adjacent to a temporal edge with label $5$, and there is no $5$-labeled path from $a_i$ to $p$.

\item[Edge between $b_i$ and $q$ is missing for some $i$:]
Then there is no temporal path from $q$ to $b_i$, since $b_i$ is only adjacent to a temporal edge with label $1$, and there is no $1$-labeled path from $q$ to $b_i$.

\item[Edge between $p$ and $t\in T$ is missing:]
Then there is no temporal path from $p$ to $t$: Going from $p$ to $A$ or $x$ doesn't help, since we then have to use label $5$, which is a dead-end. Using another edge with label $4$ is also a dead-end.

\item[Edge between $q$ and $t\in T$ is missing:]
Then there is no temporal path from $t$ to $q$: Reasons similar as above.

\item[Edge between $p$ and $x$ or edge between $q$ and $x$ is missing:]
Then there is no temporal path from $x$ to $p$, or from $q$ to $x$ respectively.
\end{description}

So, we have shown that all temporal edges not present in the Steiner Tree instance are necessary for a feasible solution to the $\TC$ instance. Furthermore, every temporal path between two vertices of $S$ can consist only of temporal edges with label $3$: If an edge between $S$ and $A$ is traversed, then we can no longer proceed, since it has label $5$. If an edge between $S$ and $B$ is traversed, then an edge between $B$ and $q$ must be traversed, but then we can no longer proceed to $S$, since we have traversed an edge with label $5$. This means that the edges with label $3$ in some feasible solution define a feasible solution to the Steiner Tree instance.

Let us now show that taking all the edges not in the Steiner Tree instance together with the set of temporal edges corresponding to a feasible solution to the Steiner Tree instance results in a temporally connected subgraph of $\G$. For this, it suffices to prove that the vertices in $S \cup T \cup \{p\} \cup \{q\}$ are connected to each other with temporal paths, since the rest of the vertices are connected to them with both minimum and maximum labels ($1$ and $5$). Hence, temporal connectivity is implied. We already have that there is a $3$-labeled path between each pair of vertices in $S$. Consider the following cases:
\begin{description}
\item[From $u_i\in S$ to $t\in T$:] Move from $u_i$ to some $t'\in T$ with label $3$, and then to $t$ with label $4$ via $p$. (This also connects $u_i$ to $p$, and $t'$ to $t$).

\item[From $u_i\in S$ to $q$:] $(u_i,b_i,q)$

\item[From $t\in T$ to $u_i \in S$:]
Move from $t$ to some $t'$ (with label $2$ via $q$) that is connected to $u_j\in S$ with label $3$, and then move on the $3$-labeled path to $u_i$. (This also connects $t$ to $q$, $q$ to $t$ and $q$ to $u_i$).

\item[From $t\in T$ to $p$:]
Move on edge with label $4$. (This also connects $p$ to $t$).

\item[From $p$ to $u_i\in S$:] $(p, a_i, u_i)$

\item[From $p$ to $q$:] $(p,x,q)$

\item[From $q$ to $p$:] $(q, t, p)$, for some $t \in T$. 
\end{description}

Since we have shown all the required facts, the claim is proven.
\end{proof}

%% file: appendix_special.tex
\section{Temporal Connectivity on Trees and Cycles: Omitted Proofs}
\label{s:app:special}

\subsection{Temporal Connectivity on Trees: A Polynomial-Time Algorithm}
\label{s:app:tree}

$\TC$ can be solved exactly in polynomial time if the underlying graph is a tree.

\begin{lemma}\label{l:tree}
Let $\G$ be a temporal tree. Then, there is an optimal solution to $\TC$ on $\G$ that uses at most two time labels of each edge.
\end{lemma}

\begin{proof}
Let $T$ be the underlying tree of $\G$. For some edge $e$, denote by $T_A$, $T_B$ the connected components in which the $T$ is partitioned after the deletion of $e$. For some fixed node $u$ in $T_A$ and every node $v$ in $T_B$, there is a temporal path in the optimal solution that connects $u$ to $v$.
Among all these temporal paths, let $(u, (e_1,t_1), (e_2,t_2), \ldots, (e,t), \ldots, v)$ be the path with the minimum traversal time for edge $e$. Note that for all nodes $v'$ of $T_B$ other than $v$, there exists a path $(u, (e_1',t_1'), (e_2',t_2'), \ldots, (e,t'), p', v')$ from $u$ to $v'$.
Since $t'\geq t$, the path $(u, (e_1,t_1), (e_2,t_2), \ldots, (e,t), p', v')$ is also a valid temporal $u - v'$ path. This means that in some optimal solution to $\TC$, one time label $t(e, u)$ on edge $e$ can be used by all paths from $u$ to vertices in $T_B$. If we keep as the time label of $e$ the maximum time label $t(e, v)$, over all vertices $v \in T_A$, this can be used by any temporal path from some vertex in $T_A$ to some vertex in $T_B$. Hence, for all paths from $T_A$ to $T_B$, one label suffices for $e$.
By symmetry, the same holds for all paths from $T_B$ to $T_A$. Overall, there is an optimal solution to $\TC$ where at most two time labels are used for any edge $e$.
\end{proof}

Now, $\TC$ on temporal trees can be solved as follows: We root the tree arbitrarily at some vertex $r$. For some vertex $u$, let $f(u,t_i,t_o)$ denote the minimum cost of $\TC$ on the temporal tree $T_u$
rooted at $u$, with the additional constraint that for every node $v$ in $T_u$, there must exist a temporal path from $u$ to $v$ using only temporal edges with time label at least $t_i$ and a temporal path from $v$ to $u$ using only temporal edges with label at most $t_o$. Denote by $c_j$, for $i\in [1,x_u]$ the children of $u$ in the tree, where $x_u$ is the number of children. Also, let  $g(u,j,t_i,t_o)$ denote the minimum cost for $\TC$ on the temporal tree rooted at $u$ but with no children other than $c_1,\ldots,c_j$, with the constraints concerning $t_i, t_o$ as stated above, but restricted to the subtrees of the first $j$ children. Obviously, $g(u,0,t_i,t_o)=0$ and $g(u,x_u,t_i,t_o)=f(u,t_i,t_o)$. We now observe that
\begin{align*}
 g(u, j, t_i, t_o) = \min_{t'_i, t'_o \in \nats}
      & \big\{\,g(u, j-1, \max(t_i,t'_o), \min(t_o,t'_i)) +
                 f(c_j, t'_i, t'_o) + q(u, c_j, t'_i, t'_o)\,: \\
      &\ \ t'_i \geq t_i \mbox{ and } t'_o \leq t_o\,\big\}
\end{align*}
In the equation above, $t'_i$ denotes the time that the edge $\{u,c_j\}$ is traversed from $u$ to $c_j$ and $t'_o$ denotes the time that $\{u,c_j\}$ is traversed from $c_j$ to $u$. Moreover,
\[
 q(u, c_j, t_i', t_o') = \left\{\begin{array}{ll}
  w((u,c_j),t_i') + w((u,c_j),t_o') \ \ \ & \mbox{if $t_i'\neq t_o'$}\\
  w((u,c_j),t_i') & \mbox{otherwise}
  \end{array}\right.
 \]

If $L$ is the total number of temporal edges, then $f$ and $g$ can be computed in time $O(n L^4)$ using dynamic programming. The optimal value of $\TC$ is then $f(r,0,\infty)$.
In order to get the actual $\TC$ solution, we just follow the parent links between states of the recurrence relation.

\subsection{Temporal Connectivity on Cycles: 2-Approximation Algorithm}
\label{s:app:cycle}

When the underlying graph is a cycle, we have a 2-approximation algorithm for $\TC$.
In order for the temporal graph to be connected, for every vertex in the cycle there must exist temporal paths to every other vertex.
Specifically, if we mark with $v_0,v_1, v_2, \ldots, v_{n-1},v_n\equiv v_0$ the nodes of the cycle, then for
a solution to be feasible,
for every node $v_i$ there must
exist a node $v_j$, such that there are temporal paths $v_i,v_{i+1},\ldots,v_{j-1},v_j$ and $v_i,v_{i-1},\ldots,v_{j+2},v_{j+1}$ in the solution
(vertex indices are taken modulo $n$). See also Figure~\ref{fig:cycle}.a.
\begin{figure}[t]
	\begin{minipage}[t]{.45\textwidth}\centering
		\includegraphics[width=.6\textwidth]{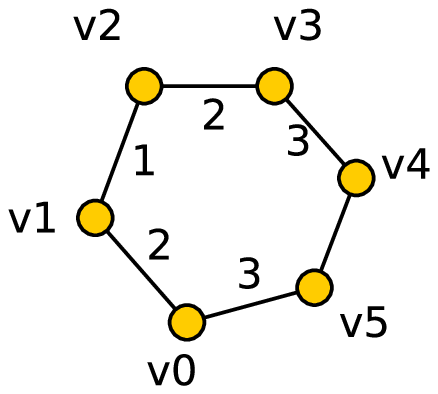}
		\caption*{(a) $v_1$ is connected to any other vertex through temporal
			paths $v_1,v_2,v_3,v_4$ and $v_1,v_0,v_5$.}
	\end{minipage}\hfill%
	\begin{minipage}[t]{.45\textwidth}\centering
		\includegraphics[width=.8\textwidth]{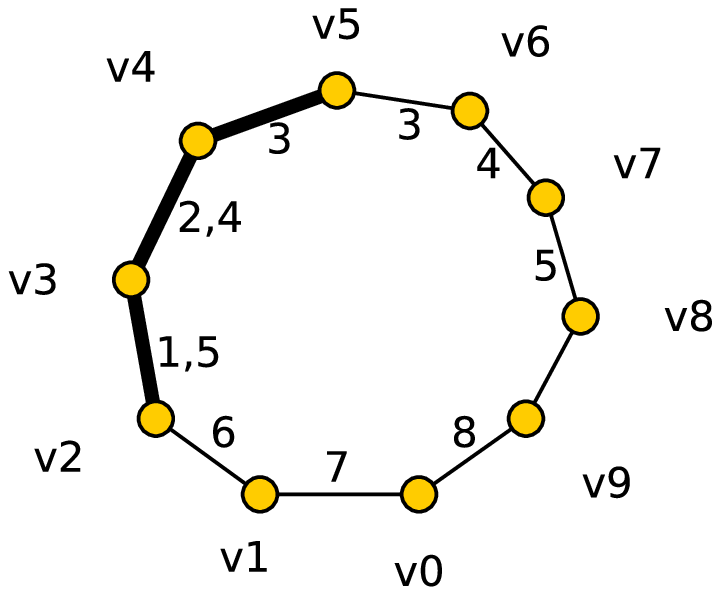}
		\caption*{(b) The sector $v_2,v_3,v_4,v_5$. Its increasing and decreasing paths are
			$v_2,v_3,\ldots,v_8$ (using labels $1,2,3,3,4,5$) and $v_5,v_4,\ldots,v_9$
				(using labels $3,4,5,6,7,8$) respectively.}
	\end{minipage}
	\caption{Connectivity in temporal cycles.}
	\label{fig:cycle}
\end{figure}

\ifx true false
\begin{center}
\begin{tikzpicture}
        \node[draw,circle,minimum width=2cm] at (0,0) (A) {};
        \draw[>=latex,->] (A.west) node[anchor=east]{$v_i$} to [in=180, out=90] (0,1.5) to [in=90,out=0](A.north east);
        \draw[>=latex,->] (A.west)
	to [in=180, out=-90] (0,-1.5) to [in=225,out=0] (1.2,-0.9) to [in=0,out=45] (A.north east) node[anchor=south west]{$v_j$};
\end{tikzpicture}
\end{center}
\fi

\begin{lemma}
Suppose that in some temporally connected subgraph $\G'$ of the temporal cycle, there is a temporal path $v_i,v_{i+1},\ldots,v_j$ and for some $k\in \{i,i+1,\ldots,j\}$ there is also a temporal path $v_k,v_{k-1},\ldots,v_{j+2},v_{j+1}$. Then every vertex $v_p \in \{i,i+1,\ldots,k\}$ is connected to every other vertex in $\G'$.
\end{lemma}

\begin{proof}
Since the temporal path $v_i,v_{i+1},\ldots,v_j$ exists in $\G'$, so does the path $v_p,v_{p+1},\ldots,v_j$. Also, since the path $v_k,v_{k-1},\ldots,v_{j+2},v_{j+1}$ exists in $\G'$, so does the path $v_p,v_{p-1},\ldots,v_{j+2},v_{j+1}$. So there are paths from $v_p$ to every other vertex in $\G'$.
\end{proof}

This leads us to the following definition (see also Figure~\ref{fig:cycle}.b for an example).

\begin{definition}
Consider a temporal cycle $\G$ and a temporally connected subgraph $\G'$ of $\G$. A \textit{sector}
of $\G$ with respect to $\G'$ is a contiguous sequence of vertices $(v_i, v_{i+1}, \ldots ,v_j)$ (it may be $i = j$) such that there exists a vertex $v_k \notin \{i,i+1,\ldots,j-1\}$ and the temporal paths
$p_{\mathrm{incr}} = (v_i, v_{i+1}, \ldots, v_k)$ and $p_{\mathrm{decr}} = (v_{j}, v_{j-1}, \ldots, v_{k+1})$ are present in $\G'$. We refer to $p_{\mathrm{incr}}$ as the \textit{increasing path} and to $p_{\mathrm{decr}}$ as the \textit{decreasing path} of the sector $(v_i, v_{i+1}, \ldots ,v_j)$. The \textit{cost} of a sector $S = (v_i, v_{i+1}, \ldots ,v_j)$ is the minimum, over all choices of increasing and decreasing paths pairs for $S$ (in an increasing and decreasing path pair, if the increasing path ends at $k$, the decreasing ends at $k+1$), of the total cost of all temporal edges present in these increasing and decreasing paths.
\end{definition}

Note that the vertices of a temporal cycle $\G$ with a temporally connected subgraph $\G'$ can always be partitioned into sectors with respect to $\G'$.

\begin{lemma}\label{l:cycle}
For a temporal cycle $\G$, there exists an optimal $\TC$ solution $\G'$, such that no two sectors' increasing (resp. decreasing) paths use the same temporal edge (the same edge at the same time).
\end{lemma}

\begin{proof}
Suppose that $\G'$ is partitioned to $s$ disjoint sectors and that $s$ is minimum among all optimal solutions. Furthermore, suppose that there are two (disjoint) sectors, $S_0: v_0,v_1,\ldots,v_j$ and $S_1: v_k,v_{k+1},\ldots,v_l$, defined by this partition, whose increasing paths contain a common temporal edge. Note that disjointness implies $j<k$ and $l<n(\G)$. We will show that this contradicts the minimality of $s$. Suppose the increasing path of sector $S_0$ is $v_0, v_1, \ldots, v_a$ and the increasing path of sector $S_1$ is $v_k,v_{k+1}, \ldots, v_b$.
We distinguish between two cases:

\begin{description}
\item[{$a \in [j, k-1]$}:] In order for the two increasing paths to have a common temporal edge, the increasing path of $S_1$ must be of the form $v_k,v_{k+1}, \ldots , v_z,v_{z+1}, \ldots, v_m$, where $z \in \{0,1,\ldots,a-1\}$ and the edge between $v_z$ and $v_{z+1}$ is visited at the same time by the two increasing paths. This means that we can merge these two paths, creating a path $v_k,v_{k+1},\ldots,v_a$. Combining this with the decreasing path of $S_0$, which is $v_j, v_{j-1}, \ldots, v_{a+1}$, this means that we can replace all sectors that contain vertices $v_k, v_{k+1} \ldots , v_j$ with a single sector consisting of exactly those vertices. Additionally, we have not added extra edges so the cost of the solution will not increase. Also note that the case $b\in [l,n(\G)-1]$ is symmetric to the above case.

\begin{center}
\begin{tikzpicture}
\node[draw,circle,minimum width=2cm] at (0,0) (A) {};
\draw [line width=4pt,domain=160:200] plot ({cos(\x)}, {sin(\x)}) node[anchor=south east]{$S_0$};
\draw [line width=4pt,domain=-20:20] plot ({cos(\x)}, {sin(\x)}) node[anchor=north west]{$S_1$}
	node[anchor=east]{$v_k$};

\draw[>=latex,->] ({cos(200)},{sin(200)}) node[anchor=west]{$v_0$} to [in=180, out=-220] (0,1.5) to [in=90,out=0](A.north east)
	node[anchor=south west]{$v_a$};
\draw[>=latex,->] ({cos(160)},{sin(160)}) node[anchor=west]{$v_j$}
	to [in=180, out=220] (0,-1.5)
	to [in=225,out=0] (1.2,-0.9) to [in=0,out=45] (A.north east);
\end{tikzpicture}
\end{center}

\item[{$a\in [k,n(\G)-1]$ and $b\in [0,k-1]$}:]
Let the common temporal edge of the two paths be connecting $v_z$ and $v_{z+1}$. If $z\in \{0,1,\ldots,k-1\}$, then merging the two increasing paths we obtain a new path
$v_k, v_{k+1},\ldots,v_z,v_{z+1},\ldots,v_k, \ldots, v_a$, which
means that we can replace all sectors with a single sector
with no extra cost. As for the case $z\in \{k,k+1,\ldots,a-1\}$, we can combine temporal
paths $v_0,v_1,\ldots,v_a$ with $v_l,v_{l-1},\ldots,v_{a+1}$ to merge
vertices $v_0,v_1,\ldots,v_l$ into a single sector with no extra cost.

\begin{center}
\begin{tikzpicture}
\node[draw,circle,minimum width=2cm] at (0,0) (A) {};
\draw [line width=4pt,domain=160:200] plot ({cos(\x)}, {sin(\x)}) node[anchor=south east]{$S_0$};
\draw [line width=4pt,domain=-20:20] plot ({cos(\x)}, {sin(\x)}) node[anchor=north west]{$S_1$}
	node[anchor=east]{$v_k$};

\draw[>=latex,->] ({cos(200)},{sin(200)})  node[anchor=west]{$v_0$} to [in=180, out=-220] (0,1.5)
to [in=135, out=0] (1.2,0.9) to [in=0,out=-45] (A.south east)
	node[anchor=north west]{$v_a$};
\draw[>=latex,->] ({cos(160)},{sin(160)}) node[anchor=west]{$v_j$}
to [in=180, out=220] (0,-1.5) to [in=-90,out=0](A.south east);
\end{tikzpicture}
\end{center}

\end{description}
\end{proof}

The above lemma implies that in some $\TC$ solution, every temporal edge is used at most twice by different sectors.
Thus, ignoring the mutual dependence between sectors we can get a $2$-approximate solution to our problem.

\begin{theorem}
There is an algorithm that runs in polynomial time and outputs a feasible $\TC$ solution of a temporal cycle $\G$, with cost at most $2$ times
the optimal.
\end{theorem}
\begin{proof}
Denote the vertices of the cycle as $v_0,v_1,\ldots,v_n\equiv v_0$.
First, we can try all possible vertices as the start of a sector. Let this be vertex $v_0$ of the cycle.
Let $p(i)$ denote the partition of $v_0,v_1,\ldots,v_i$ into sectors with minimum cost (the cost of a set of disjoint sectors is defined as
the sum of costs of the individual sectors) and $c(i)$ its cost.

Furthermore, define as
$\mathrm{inc\_path}(i,j)$ the minimum cost of a temporal path that visits $v_i,v_{i+1},\ldots,v_j$ in this order and
$\mathrm{dec\_path}(i,j)$ the minimum cost of a temporal path that visits $v_i,v_{i-1},\ldots,v_j$ in this order.
$\mathrm{inc\_path}$ and $\mathrm{dec\_path}$ can be precomputed in time $O(n M)$ and space $O(n^2+M)$, where $M$ is the number of temporal edges.

Also, define $\mathrm{sector\_cost}(i,j)$ to be the cost of the sector $v_i,v_{i+1},\ldots,v_j$. This can be precomputed in time $O(n^3)$ and space $O(n^2)$, as

\begin{align*}
 \mathrm{sector\_cost}(i,j) =
   \min_{k\in\{0,\ldots,n-1\}\setminus\{i,i+1,\ldots,j-1\}} \Big\{\,\mathrm{inc\_path}(i,k)+\mathrm{dec\_path}(j,k+1)\,\Big\}
\end{align*}

\ifx true false
\begin{align*}
 f(i,a_1,t_1,\ldots,a_p,t_p) =
   \min_{t_{p+1}\in [L+k]}&\!\left\{ f(c_1 (i),a_1,t_1,\ldots,a_{j'},L+p+1,\ldots,a_p,t_p,0,t_{p+1})\right.\,:\\
      & \ \ \ t_{p+1}\leq L\lor t_{p+1}-L=j\in [p]\Bigg\}
\end{align*}
\fi

Then, we have that $c(0)=\mathrm{sector\_cost}(0,0)$
and that for $i>0$,
\begin{align*}
  c(i) =
    \min \!\left\{\mathrm{sector\_cost}(0,i), \min_{j\in \{0,\ldots,i-1\}} \{c(j) + \mathrm{sector\_cost}(j+1,i)\}\right\}
\end{align*}
$p(i)$ can also be easily computed this way.

Finally, for the choice of $v_0$ that achieves the minimum $c(n-1)$,
the algorithm outputs the union of the optimal increasing and decreasing paths of the sectors defined in $p(n-1)$.

The temporal graph that the algorithm produces is a feasible solution, because every vertex belongs to some sector and is thus
connected to every other vertex.

To show that it is $2$-approximate of the optimal solution, suppose that we double every temporal edge of the graph. The optimal cost
$OPT$ will not change.
Yet, we can transform an optimal solution in which every temporal edge is used by at most $2$ sectors to a solution in which every temporal
edge is used by at most $1$ sector in the doubled graph, with cost at most $2\cdot OPT$. Either this or a smaller cost solution will be
considered and produced by our algorithm.
Overall, our algorithm uses $O(n^3+n\cdot M)$ time and $O(n^2 + M)$ space.
\end{proof}